\documentclass[10pt,conference]{IEEEtran}
\IEEEoverridecommandlockouts
\usepackage{cite}
\usepackage{amsmath,amssymb,amsfonts}
\usepackage[linesnumbered, ruled]{algorithm2e}
\usepackage{graphicx}
\usepackage{textcomp}
\usepackage{xcolor}
\usepackage{amsthm}
\usepackage{enumerate}
\usepackage{scalerel}
\usepackage{array}
\usepackage[hidelinks]{hyperref}
\usepackage{subfigure}
\usepackage{bbm}
\usepackage{dsfont}
\usepackage{url}
\usepackage{booktabs}
\usepackage{makecell}
\usepackage[table]{xcolor}

\def\BibTeX{{\rm B\kern-.05em{\sc i\kern-.025em b}\kern-.08em
    T\kern-.1667em\lower.7ex\hbox{E}\kern-.125emX}}

\newtheorem{theorem}{Theorem}

\newtheorem{proposition}{Proposition}
\newtheorem{definition}{Definition}

\DeclareRobustCommand*{\IEEEauthorrefmark}[1]{%
  \raisebox{0pt}[0pt][0pt]{\textsuperscript{\footnotesize #1}}%
}

\begin{document}

\title{Fault-Tolerant Offloading Framework for Real-Time Applications in Mobile Edge Computing
\thanks{This work was supported by the MoE Tier-2 grant MOE-T2EP20224-0007. The work was also supported by DesCartes, National Research Foundation, Prime Minister’s Office, Singapore under its Campus for Research Excellence and Technological Enterprise (CREATE) programme.}
}

\author{
\IEEEauthorblockN{Chuanchao Gao\IEEEauthorrefmark{1,2}
, Yiyang Gao\IEEEauthorrefmark{1}, Michael Yuhas\IEEEauthorrefmark{1,2} Arvind Easwaran\IEEEauthorrefmark{1,2}}
\IEEEauthorblockA{\IEEEauthorrefmark{1}\textit{College of Computing and Data Science}}
\IEEEauthorblockA{\IEEEauthorrefmark{2}\textit{Energy Research Institute @ NTU, Interdisciplinary Graduate Programme}}
\IEEEauthorblockA{
\textit{Nanyang Technological University, Singapore}\\
gaoc0008@e.ntu.edu.sg, yiyang009@e.ntu.edu.sg, 
michaelj004@e.ntu.edu.sg, arvinde@ntu.edu.sg}
}

\maketitle

\begin{abstract}
This paper studies a Deadline-constrained task Offloading and resource Allocation Problem (DOAP) in a backhaul-network-enabled Mobile Edge Computing (MEC) system with both bandwidth and computational resource constraints. The goal is to maximize the total inference accuracy improvement across all tasks. We formulate this problem, denoted as $\mathbf{P}$, as an Integer Linear Programming (ILP) problem and show it is NP-hard via a reduction from the maximum-weight three-dimensional matching problem. To address $\mathbf{P}$, we develop $\mathtt{FastSA}$, a linear-time $\frac{1}{10}$-approximation algorithm based on the local-ratio technique. To the best of our knowledge, $\mathtt{FastSA}$ provides the first constant-factor performance guarantee for backhaul-network-enabled DOAPs. In addition, we introduce a fault detection and recovery framework that enables robust MEC operation under backhaul link and edge server failures. Experiments using profiled image classification workloads and real-world taxi traces show that $\mathtt{FastSA}$ achieves an average accuracy improvement of $25.6\%$ over benchmark algorithms under fault-free conditions while maintaining the lowest runtime. Under fault-injected scenarios, our framework reduces performance degradation by an average of $12.8\%$, demonstrating its effectiveness in fault handling.
\end{abstract}


\addtolength{\textfloatsep}{-3mm}

\section{Introduction}\label{sec:introduction}
Internet of Things (IoT) has emerged as a forefront technology (e.g., Cisco predicts 500 billion IoT devices by 2030 \cite{ji2021guest}), driven by continuous advancements in hardware, software, and communication technologies (e.g., ultra-reliable low-latency communication features of 5G)~\cite{gao2024energy}. 
Meanwhile, the development of Artificial Intelligence (AI) has spurred growing interest in deploying AI services on IoT devices, enabling applications such as path planning and object detection in autonomous vehicles~\cite{yurtsever2020survey}. These applications are typically computation-intensive and subject to strict latency requirements, posing significant challenges for resource-constrained IoT devices. To achieve reasonable execution time and energy efficiency, model compression techniques such as quantization~\cite{polino2018model} are often applied when deploying AI services on IoT devices, albeit at the cost of reduced inference accuracy.

To obtain high-quality inference results while satisfying the latency requirements of End Devices (EDs), Mobile Edge Computing (MEC) has emerged as a promising paradigm. In MEC, Edge Servers (ESs) are deployed at the network edge to provide wireless bandwidth for task offloading and computational resources for task processing. A task generated by an ED can be offloaded to a nearby ES (termed offload-ES), which is then processed on the offload-ES or further forwarded to another ES for processing (termed process-ES) via the wired backhaul network. By leveraging ESs, MEC enables the use of full-scale AI models for high-quality inference while keeping computation close to EDs. This proximity substantially reduces communication latency compared to cloud computing, making MEC well-suited for latency-sensitive applications.

Despite the advantages of MEC, the bandwidth and computational capacity of ESs can become bottlenecks as the number of tasks increases, necessitating efficient strategies for \textit{task offloading} (i.e., determining the offload-ES and process-ES for each task) and \textit{resource allocation} (i.e., deciding the amount of bandwidth and computational resource to allocate for each task). Jointly optimizing task offloading and resource allocation enables adaptive resource management based on both ES resource availability and the wireless channel quality.

Numerous studies have explored Deadline-constrained task Offloading and resource Allocation Problems (DOAP), which can be broadly categorized into two groups. The first group (e.g., \cite{dai2018joint, zhu2019folo, xu2020energy, dai2020edge, xie2022providing, li2023throughput, qi2024joint, song2024cross, liao2024an, he2024joint, hossain2023network, jiang2024aerial, wang2019joint, li2019cooperative, thai2021optimal, fan2023joint, lin2023a, gao2025local}) assumes that each offloaded task is processed on its selected offload-ES. The second group (e.g., \cite{hu2020dynamic, zhong2025joint, li2024a, du2018computation, gao2022deadline, wang2024joint, gao2025energy}) considers a more flexible setting where tasks being offloaded can be further forwarded to another ES for processing via the backhaul network; we term this group as forwarding-enabled DOAP. While the first group simplifies system design, the second offers enhanced load-balancing capability by leveraging inter-ES cooperation. The limited coverage of 5G signals further emphasizes the value of backhaul-assisted task forwarding. When requests crowd around a single ES, resource congestion may occur, whereas other ESs with available capacity remain underutilized due to wireless network coverage constraints. Despite its advantage, the forwarding-enabled DOAP remains underexplored, and most existing studies rely on heuristic approaches without theoretical performance guarantees \cite{hu2020dynamic, zhong2025joint, li2024a, du2018computation, gao2022deadline, wang2024joint}. This gap motivates the design of efficient algorithms that can achieve both computational efficiency and solution quality.

Recently, a $\frac{1-\alpha}{2}$-approximation algorithm was proposed for the forwarding-enabled DOAP \cite{gao2025energy}, ensuring that the achieved utility value is at least $\frac{1-\alpha}{2}$ of the optimal. Here, $\alpha$ represents the maximum fraction of an ES’s capacity that can be allocated to a single task. Although this algorithm provides a theoretical bound, its performance guarantee depends on $\alpha$, leading to large variation under practical scenarios with different $\alpha$ requirements. Furthermore, most existing studies assume a fault-free MEC environment, overlooking the potential impact of backhaul network or ES failures. In real deployments, failures in backhaul links or ESs are inevitable and may cause severe service disruptions in the absence of effective fault-handling mechanisms. To the best of our knowledge, no prior work on DOAP has explicitly investigated fault detection and recovery in MEC systems.

In this paper, we investigate a forwarding-enabled DOAP in MEC, considering both bandwidth and computational resource constraints, i.e., the total resource that each ES can allocate to tasks is bounded by its resource capacity for any resource type. The objective is to maximize the total accuracy improvement across all tasks. We denote this problem as $\mathbf{P}$ and formulate it as an Integer Linear Programming (ILP) problem. The maximum weight three-dimensional matching problem \cite{burkard2009assignment}, which is NP-hard, can be reduced to a special case of $\mathbf{P}$ where at most one task can be offloaded to or processed on each ES. Hence, $\mathbf{P}$ is also NP-hard.

To address $\mathbf{P}$, we propose a linear-time $\frac{1}{10}$-approximation algorithm, $\mathtt{FastSA}$, based on a local-ratio technique~\cite{bar2004local}. To the best of our knowledge, it is the first constant-factor approximation algorithm developed for forwarding-enabled DOAPs. In addition, we design a distributed network synchronization framework similar to the Open Shortest Path First (OSPF) Protocol \cite{ospfv2moy}, where each ES periodically monitors its connections with neighboring ESs and broadcasts any detected changes. This framework enables the system to promptly detect and recover from failures in backhaul links and ESs, thereby ensuring service continuity in MEC systems. The main contributions of this paper are summarized as follows:
\begin{itemize}
    \item We address a forwarding-enabled DOAP, $\mathbf{P}$, for real-time task offloading in MEC with both bandwidth and computational resource constraints. To maximize the total accuracy improvement across all tasks, we propose a linear-time $\frac{1}{10}$-approximation algorithm, $\mathtt{FastSA}$, based on the local-ratio technique, providing the first constant-factor theoretical guarantee for forwarding-enabled DOAPs.
    \item We design a distributed backhaul network synchronization framework to detect and recover from failures in backhaul links and ESs. To validate its effectiveness, we extend the open-source simulator mecRT \cite{mecrt} with this framework, creating the first MEC simulator that supports fault-detection and recovery for real-time task offloading.
    \item We evaluate $\mathtt{FastSA}$ using profiled image classification workloads and a real-world taxi trace dataset~\cite{taxi-data} under fault-free and fault-injected MEC simulations. Results show that the average accuracy improvement of $\mathtt{FastSA}$ surpasses benchmark algorithms by $25.6\%$ in fault-free scenarios while maintaining the lowest scheduling overhead. Under fault conditions, the algorithm performance degradation reduces $12.8\%$ on average when the proposed framework is enabled, showing its effectiveness in maintaining service robustness when a fault occurs.
\end{itemize}

\section{Related Works}\label{sec:related-work}

\begin{table}
    \caption{Category of studies on deadline-constrained task offloading and resource allocation}
    \centering
    \resizebox{\columnwidth}{!}{
    \begin{tabular}{|l|c|c|}
    \hline
    
    & 
    \thead{Decide where to \\ offload } 
    & 
    \thead{Decide where to\\ offload and process}
    \\ \hline
    
    \thead{Allocation for $\mathbf{B}$ or $\mathbf{C}$ Unspecified\\ Contention on $\mathbf{B}$ or $\mathbf{C}$ }
    & 
    \thead{ \cite{dai2018joint, zhu2019folo, xu2020energy, dai2020edge, xie2022providing, li2023throughput, qi2024joint, song2024cross, liao2024an}$^{\lozenge \heartsuit}$, \\ \cite{he2024joint}$^{\blacklozenge \spadesuit}$}  
    &  
    \thead{\cite{hu2020dynamic, zhong2025joint}$^{\lozenge \heartsuit}$}
    \\ \hline

    \thead{Allocation for $\mathbf{B}$ or $\mathbf{C}$ Unspecified\\ Contention on $\mathbf{B}$ and $\mathbf{C}$} 
    &
    \thead{\cite{hossain2023network, jiang2024aerial}$^\lozenge$}  
    & 
    \thead{\cite{li2024a}$^{\blacklozenge}$}
    \\ \hline
    
    \thead{Allocation for $\mathbf{B}$ and $\mathbf{C}$ Unspecified\\ Contention on $\mathbf{B}$ and $\mathbf{C}$} 
    & 
    \thead{\cite{wang2019joint, li2019cooperative, thai2021optimal, fan2023joint, lin2023a, gao2025local}}
    & 
    \cellcolor{blue!25}{\cite{du2018computation, gao2022deadline, wang2024joint, gao2025energy}}
    \\ \hline
    \end{tabular}}
    \vspace{1ex}
    
    {\raggedright \footnotesize $\mathbf{B} \leftarrow$ bandwidth;
    $\blacklozenge \leftarrow$ unspecified allocation for $\mathbf{B}$; $\spadesuit \leftarrow$ contention on $\mathbf{B}$;\\ 
    $\mathbf{C} \leftarrow$ comp. res.; $\lozenge \leftarrow$ unspecified allocation for $\mathbf{C}$; $\heartsuit \leftarrow$ contention on $\mathbf{C}$; \\
    \par}
    \label{tab:related-work}
\end{table}

Due to the promising potential of MEC, problems on task offloading and resource allocation in MEC have been extensively investigated~\cite{jamil2022resource, luo2021resource, ramanathan2020survey}. These problems are typically categorized into two types: deadline-constrained and deadline-free (e.g., response-time minimization). Depending on whether resource allocation is pre-specified, they can be further divided into (i) problems determining task offloading under fixed resource allocation and (ii) problems jointly optimizing task offloading and resource allocation. This section focuses on studies addressing deadline-constrained task offloading and resource allocation problems (DOAP).

The rows of Table~\ref{tab:related-work} classify existing DOAP studies based on resource contention and unspecified resource types. Studies considering single-resource contention mainly focused on computational resources with unspecified computational resource allocation~\cite{dai2018joint, zhu2019folo, xu2020energy, dai2020edge, xie2022providing, li2023throughput, qi2024joint, song2024cross, liao2024an, hu2020dynamic, zhong2025joint}, whereas others examined bandwidth contention with unspecified bandwidth allocation~\cite{he2024joint}. Among studies addressing joint bandwidth and computational resource contention, some assumed one unspecified resource type~\cite{hossain2023network, jiang2024aerial, li2024a}, while others considered both unspecified~\cite{wang2019joint, li2019cooperative, thai2021optimal, fan2023joint, lin2023a, gao2025local, du2018computation, gao2022deadline, wang2024joint, gao2025energy}.

The columns of Table~\ref{tab:related-work} further classify DOAP studies by offloading type. In the first column~\cite{dai2018joint, zhu2019folo, xu2020energy, dai2020edge, xie2022providing, li2023throughput, qi2024joint, song2024cross, liao2024an, he2024joint, hossain2023network, jiang2024aerial, wang2019joint, li2019cooperative, thai2021optimal, fan2023joint, lin2023a, gao2025local}, each task can be offloaded to at most one nearby ES (multiple tasks can be offloaded to the same ES under bandwidth constraints) and is processed on the offload-ES. Thus, the decision only involves selecting the offload-ES. The second column~\cite{hu2020dynamic, zhong2025joint, li2024a, du2018computation, gao2022deadline, wang2024joint, gao2025energy} represents forwarding-enabled MEC, where ESs are interconnected via a wired backhaul network and an offloaded task can be further forwarded to another ES for processing. Hence, both the offload-ES and process-ES must be jointly determined.

This paper addresses the DOAP under the most general setting (i.e., the blue category in Table~\ref{tab:related-work}). A few studies have investigated this setting~\cite{du2018computation, gao2022deadline, wang2024joint, gao2025energy}. Among them, some~\cite{du2018computation, wang2024joint} assumed each task has a fixed offload-ES, whereas others~\cite{gao2022deadline, gao2025energy} allowed multiple offloading candidates per task. To solve these problems, heuristic approaches without theoretical guarantees were commonly adopted~\cite{du2018computation, gao2022deadline, wang2024joint}, while only a few works~\cite{gao2025energy} provided theoretical performance bounds. Specifically, Gao \emph{et al.}~\cite{gao2025energy} developed a graph-matching-based algorithm achieving a theoretical bound of $\frac{1-\alpha}{2}$, where $\alpha$ denotes the maximum fraction of an ES’s capacity that can be allocated to a single task. However, this bound depends on $\alpha$, resulting in variation in the performance guarantee.

In contrast, this paper considers a setting where each task has multiple offloading candidates and proposes a linear-time approximation algorithm with a constant performance guarantee of $1/10$. To the best of our knowledge, this is the first constant-factor approximation algorithm for forwarding-enabled DOAPs. Moreover, none of the existing studies in Table~\ref{tab:related-work} address fault-tolerant task offloading. To fill this gap, we design a detection and recovery framework that handles failures in both ESs and the backhaul network.

\section{System Model and Problem Formulation}\label{sec:system-architecture}
The MEC consists of a set of logical Edge Servers (ESs) and End Devices (EDs), as illustrated in Fig.~\ref{fig:system1}. Each ES is logically associated with a 5G base station (gNB) that manages wireless bandwidth allocation during data offloading and a server that provides computational resources to host services for ED tasks. EDs and ESs communicate through the 5G network, while the ESs are interconnected via a wired backhaul network. A task can be offloaded by its ED to a nearby ES, where it is either processed locally at the associated server or forwarded to another ES through the backhaul network for processing. For ease of reading, a notation summary is provided in Table \ref{table:notation}.

\begin{figure}
    \centering
    \includegraphics[width=0.9\linewidth]{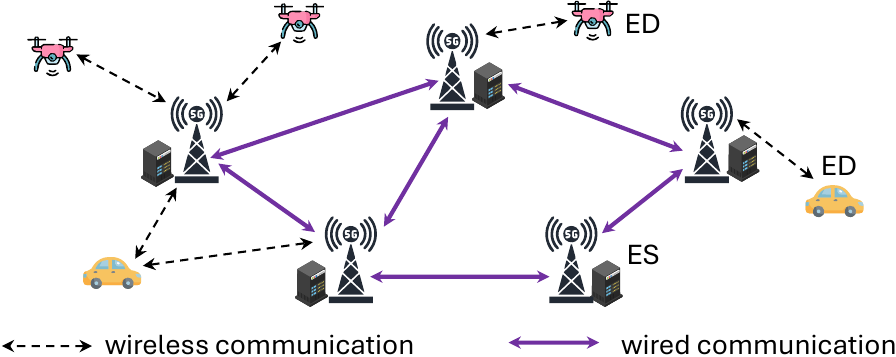}
    \caption{An Example of the mobile edge computing system.}
    \label{fig:system1}
\end{figure}

Let $\mathcal{M}$ be the set of $M$ ESs in the MEC, where $m_k \in \mathcal{M}$ denotes the $k$-th ES. The computational resources of each ES are measured in computing units (CUs), which abstract the execution capacity available on the underlying server hardware. For instance, with NVIDIA’s Green Contexts feature \cite{nvidia2025green}, the execution capacity of a GPU server can be partitioned into multiple resource units, each comprising a predefined allocation of streaming multiprocessors and memory.
Let $C_k$ be the total CUs available at $m_k$, and $C $=$ \max\{C_k | m_k \in \mathcal{M}\}$ be the maximum number of CUs that any ES can provide. 

Orthogonal Frequency Division Multiple Access, a core technology in 5G networks, divides a wireless channel into multiple orthogonal subcarriers, with every $12$ subcarriers forming a resource block, the smallest unit of bandwidth that can be allocated to users. Since PHY/MAC-layer resource block scheduling is performed by the gNB, we define a bandwidth unit (BU) as an abstraction of the PHY/MAC-layer resource block exposed to high-level control. Let $B_k$ be the total BUs available at ES $m_k$, and $B $=$ \max\{B_k | m_k \in \mathcal{M}\}$ be the maximum number of BUs that any ES can offer.

In this work, we focus on high-level offloading, forwarding, and processing decisions. Accordingly, we abstract the wireless transmission capacity available at gNBs as BUs, while direct control of PHY/MAC-layer resource block scheduling remains internal to the gNB. Similar abstraction is applied to computational resources at servers, which are represented at a higher level without exposing low-level execution scheduling.

Let $\mathcal{N}$ denote the set of $N$ tasks pending to offload, where $n_i \in \mathcal{N}$ represents the $i$-th task. Let $\mathcal{M}_i$ denote the set of ESs accessible by task $n_i$'s ED, i.e., the ESs where task $n_i$ can be offloaded. Since these tasks typically operate periodically, task $n_i$ has a period $T_i$, measured in seconds, which also serves as its implicit deadline. We focus on discriminative inference tasks (e.g., classification and detection), for which the output size is negligible compared to the input size.

For many computation-intensive tasks, local processing performance is often constrained by the hardware limitations of EDs. For example, deploying AI applications on resource-constrained EDs often requires model compression techniques to reduce computational and memory demands, usually at the cost of reduced model accuracy. In such cases, EDs can offload tasks to ESs to achieve more accurate inference results. However, the potential unavailability of ES resources due to faults or uncertainty in the wireless network necessitates the ability to process tasks locally to ensure operation safety. Let $d_i^\text{loc}$ denote the local execution time of task $n_i$; besides, we assume $d_i^\text{loc} \le T_i$ to ensure feasible local processing.

We consider two task offloading modes. In the first mode $\psi_1$, a task supports multiple execution stop points, where earlier termination yields lower-quality results (e.g., the anytime YOLO for object detection~\cite{yang2024adadet}). Task $n_i$ first executes locally for $d_i^\text{loc}$ after its release, and the intermediate data are then offloaded to the ES for continued execution to improve accuracy. The ES result is adopted if it is returned before $T_i$; otherwise, the locally obtained result is used. In the second mode $\psi_2$, the task is offloaded immediately after its release while also being executed locally in parallel. The ES result is adopted if it arrives before $T_i$; otherwise, the locally computed result is used. Let $\theta_i^1$ denote the size of the intermediate data offloaded in mode $\psi_1$, and let $\theta_i^2$ denote the size of the input data offloaded in mode $\psi_2$. Let $\Psi_i \subseteq \{\psi_1, \psi_2\}$ be the set of offloading modes that task $n_i$ can choose.

For a given task (e.g., image classification), multiple AI models may be available that perform the same function but differ in inference accuracy and execution time. Each ED is assumed to adopt an offline strategy to select a fixed local model for each task. Let $\mathcal{A}_i$ denote the set of $A_i$ candidate models available on ESs that can improve the inference accuracy of task $n_i$, and let $A $=$ \max\{A_i \mid n_i \in \mathcal{N}\}$. When task $n_i$ is granted to offload, an appropriate model is selected according to its latency requirement and the available resources. For a selected model $a \in \mathcal{A}_i$ deployed on an ES to serve $n_i$, let $u_{ia}$ denote the \textit{accuracy improvement} (gain), defined as the difference in accuracy between the model $a$ and the local model used by $n_i$. Specifically, when a job instance of $n_i$ is offloaded to an ES and processed using model $a$, \textit{and its result is returned before the deadline}, the obtained accuracy improvement for that job instance equals $u_{ia}$.

In 5G networks, each ED periodically transmits Sounding Reference Signals (SRSs) to nearby gNBs. Each gNB utilizes the received SRS to estimate the wireless channel quality between the ED and the gNB~\cite{3gpp38211}. The channel quality can be affected by various factors, including physical distance, transmission power, ED mobility, obstacles, and interference caused by other ED transmissions~\cite{kutila20205g}. Based on the estimated channel quality, each gNB selects an appropriate Modulation and Coding Scheme (MCS) to ensure reliable data transmission. The chosen MCS determines the number of bytes that a resource block can carry within a fixed time interval, thereby defining the achievable data offloading rate between an ED and the gNB, which is exposed to the high-level control plane. Let $\mu_{ik}$ denote the data offloading rate (in bytes per BU per second) when task $n_i$ is offloaded to the gNB associated with ES $m_k$, as determined by the latest SRS feedback.

Mid-band 5G gNB typically covers $500$–$1000$ meters in urban areas. An ED moving at $60$~km/h (the average vehicle speed in urban areas) travels about $1.67$ meters in $100$~ms, which is negligible compared to the gNB coverage range. Thus, consistent with prior studies~\cite{zhan2020mobility, raza2022task}, we assume that the wireless channel quality (and hence $\mu_{ik}$), as estimated from uplink SRS, remains constant for the duration in which a task’s data is offloaded (typically, of the order of $10$ms).

In this MEC, we consider a wired \textit{backhaul network} interconnecting the ESs. To support predictable and parallel data transmission across the backhaul network, we assume that deterministic Ethernet technologies, such as Time-Sensitive Networking (TSN), are available to provide bounded latency on wired links~\cite{szalay2023}. In addition, Software-Defined Networking (SDN)~\cite{alameddine2019dynamic} can be employed to configure routing paths and to provision virtual links with predefined bandwidth allocations for forwarding a task’s data. Wired backhaul links typically offer significantly higher bandwidth than wireless channels. For example, a 5G channel operating in the sub-6GHz band provides a bandwidth capacity of approximately $100$MHz~\cite{5gchannel}, whereas a commercial optical transmission system can offer bandwidths up to $4.5$THz \cite{xia2013commercial}. Given the substantial bandwidth available in wired networks, we ignore communication contention within the backhaul network, and each virtual link can be abstracted as providing a constant data transmission rate, denoted as $\beta$. Let $h_{jk}$ denote the minimum number of link hops between ESs $m_j$ and $m_k$.

\textit{Timing Model.} Suppose task $n_i$ is offloaded to ES $m_j \in \mathcal{M}_i$ with BU allocation $b_{ij} \in \{1, ..., B_j\}$, and is forwarded to ES $m_k$ for processing with CU allocation $c_{ik} \in \{1, ..., C_k\}$ and candidate model $a \in \mathcal{A}_i$.

In offloading mode $\psi_1$, task $n_i$ first executes locally for $d_i^\text{loc}$ after its job instance is released. The intermediate data $\theta_i^1$ is then offloaded to ES $m_j$. Let $d_{ij}^\text{off}$ denote the offloading time, where $d_{ij}^\text{off} = \theta_i^1/(\mu_{ij} \cdot b_{ij})$. The forwarding time from ES $m_j$ to ES $m_k$ is given as $d_{ijk}^\text{fwd} = h_{jk} \cdot \theta_i^1 /\beta $. Let $d_{ik}^\text{pro}$ denote the processing time of task $n_i$ on ES $m_k$, which is a function of the allocated CUs $c_{ik}$, the hardware type of the server associated with ES $m_k$, and the selected candidate model $a$.
Since the result data size is much smaller than the input for discriminative inference tasks, the result forwarding time from $m_k$ to $m_j$ is negligible. Besides, the result can be downloaded from $m_j$ using one BU within one time unit. Since the downlink bandwidth in 5G is generally larger than the uplink capacity and each offloaded task is guaranteed at least one BU in the uplink, we omit bandwidth contention for result downloading and define the result return time as a fixed constant $d^\text{rr}$. The deadline constraint in mode $\psi_1$ is given as
\begin{equation}\label{eq:system1}
    d_i^\text{loc} + d_{ij}^\text{off} + d_{ijk}^\text{fwd} + d_{ik}^\text{pro} + d^\text{rr} \le T_i.
\end{equation}

In offloading mode $\psi_2$, the job instance of task $n_i$ is offloaded immediately after release. Thus, $d_{ij}^\text{off} = \theta_i^2/(\mu_{ij} \cdot b_{ij})$ and $d_{ijk}^\text{fwd} = h_{jk} \cdot \theta_i^2 / \beta$. The processing time $d_{ik}^\text{pro}$ and result return time $d^\text{rr}$ remain the same as in mode $\psi_1$. The deadline constraint in $\psi_2$ is expressed as
\begin{equation}\label{eq:system2}
    d_{ij}^\text{off} + d_{ijk}^\text{fwd} + d_{ik}^\text{pro} + d^\text{rr} \le T_i.
\end{equation}

\begin{table}
    \caption{Notation (Parameters and Variables)}
    \resizebox{\columnwidth}{!}{
    \begin{tabular}[t]{cl}
        \toprule
        \multicolumn{1}{c}{symb.} & \multicolumn{1}{c}{definition} \\
        \midrule
        $\mathcal{M}$ & set of ESs, where $|\mathcal{M}| = M$ \\ 
        $C_k$ & total CUs of ES $m_k$; $C=\max\{C_k|m_k \in \mathcal{M}\}$ \\
        $B_k$ & total BUs of ES $m_k$; $B=\max\{B_k|m_k \in \mathcal{M}\}$ \\
        $\mathcal{N}$ & set of tasks, where $|\mathcal{N}|=N$\\
        $T_i$ & period (deadline) of task $n_i \in \mathcal{N}$\\
        $\mathcal{M}_i$ & set of ESs accessible by vehicle $v_i$\\
        $d_i^\text{loc}$ & local execution time of task $n_i \in \mathcal{N}$\\
        $\Psi_i$ & set of offloading modes of task $n_i$, $\Psi_i \subseteq \{\psi_1, \psi_2\}$\\
        $\theta_i^1$ & intermediate data size of $n_i$ (for offloading mode $\psi_1$)\\
        $\theta_i^2$ & input data size of $n_i$ (for offloading mode $\psi_2$)\\
        $\mathcal{A}_i$ & set of candidate models for task $n_i$ if offloaded\\
        $u_{ia}$ & improved accuracy for $n_i$ if choosing model $a \in \mathcal{A}_i$\\
        $\mu_{ik}$ & offloading rate per BU per second from $n_i$ to $m_k$\\
        $\beta$ & data rate in the virtual link of the backhaul network\\
        $h_{jk}$ & hop count in shortest path between ESs $m_j$ and $m_k$\\
        $\ell$ & $\ell \triangleq \left< n_i, m_j, m_k, b_\ell, c_\ell, \psi, a \right>$, a service instance \\ 
        $u(\ell)$ & improved accuracy by serving a task following $\ell$\\
        $\mathcal{L}$ & set of all feasible service instances\\
        $\mathcal{L}_i^\text{task}$ & $\mathcal{L}_i^\text{task} \subseteq \mathcal{L}$, set of service instances for task $n_i$ \\
        $\mathcal{L}_j^\text{off}$ & $\mathcal{L}_j^\text{off} \subseteq \mathcal{L}$, set of service instances offloaded to ES $m_j$ \\
        $\mathcal{L}_k^\text{pro}$ & $\mathcal{L}_k^\text{pro} \subseteq \mathcal{L}$, set of service instances processed on ES $m_k$ \\
        $z_\ell$ & binary selection variable for service instance $\ell$\\
        \bottomrule
    \end{tabular}}
    \label{table:notation}
\end{table}

This paper aims to determine the offloading strategy for each task $n_i$ (including offloading mode $\psi$, offload-ES $m_j$, BU allocation $b_{ij}$, process-ES $m_k$, CU allocation $c_{ik}$, and candidate model $a$) to maximize the total accuracy improvement per second. The optimization must satisfy the \textit{offloading constraints} (i.e., $\psi \in \Psi_i$, $m_j \in \mathcal{M}_i$, and $b_{ij} \le B_j$), the \textit{processing constraints} (i.e., $c_{ik} \le C_k$, $a \in \mathcal{A}_i$), the \textit{task deadline constraints} (Eq.~\eqref{eq:system1} or Eq.~\eqref{eq:system2}), and the \textit{ES resource constraints} (i.e., the total allocated resources on each ES cannot exceed its available capacity for any resource type). We denote this optimization problem as $\mathbf{P}$. 

The direct formulation of $\mathbf{P}$ is nonlinear due to deadline constraints \eqref{eq:system1} and \eqref{eq:system2}, making it difficult to solve directly. Here, we introduce a more concise ILP formulation for $\mathbf{P}$ based on service instances (defined in Definition \ref{def:instance}) of all tasks.

\begin{definition}[Service Instance]\label{def:instance}
    A service instance of task $n_i$ is defined as $\ell \triangleq \left< n_i, m_j, m_k, b_\ell, c_\ell, \psi, a \right>$, representing that task $n_i$ selects offloading mode $\psi$, and it is offloaded to ES $m_j$ with BU allocation $b_\ell$; besides, it is forwarded to ES $m_k$ for processing with CU allocation $c_\ell$ and candidate model $a$. Moreover, $\ell$ needs to satisfy (i) the offloading constraints (i.e., $\psi \in \Psi_i$, $m_j \in \mathcal{M}_i$, and $b_\ell \le B_j$), (ii) the processing constraints (i.e., $c_\ell \le C_k$, $a \in \mathcal{A}_i$), and (iii) the \textbf{deadline} constraints (i.e., Eq. \eqref{eq:system1} for $\psi=\psi_1$ or Eq. \eqref{eq:system2} for $\psi=\psi_2$).
\end{definition}

Let $\ell \triangleq \left< n_i, m_j, m_k, b_\ell, c_\ell, \psi, a \right>$ denote a service instance of task $n_i$. Given task $n_i$ and the selected model $a$, the resulting accuracy improvement per second is $u_{ia}/T_i$ (if all job instances of $n_i$ can be successfully offloaded and meet deadlines); we denote this utility value as $u(\ell)$, where $u(\ell)=u_{ia}/T_i$. Since both bandwidth and computational resource allocations are discrete, we can enumerate all feasible service instances for each task. Let $\mathcal{L}$ denote the set of all feasible service instances for all tasks in $\mathcal{N}$. We define $\mathcal{L}_i^\text{task} \subseteq \mathcal{L}$ as the subset of service instances corresponding to task $n_i$, $\mathcal{L}_j^\text{off} \subseteq \mathcal{L}$ as the subset of service instances in which the task is offloaded to ES $m_j$, and $\mathcal{L}_k^\text{pro} \subseteq \mathcal{L}$ as the subset of service instances in which the task is processed on ES $m_k$. 

The MEC comprises $N$ tasks and $M$ ESs, where each ES provides up to $B$ BUs and $C$ CUs. Each task can adopt up to two offloading modes and select from $A$ candidate models. For each bandwidth allocation, only the minimum required computational allocation needs to be determined to satisfy the task deadline (and vice versa). Consequently, the total number of service instances in $\mathcal{L}$ is at most $(2 N M^2 A \cdot \min\{B, C\})$. For each service instance $\ell \in \mathcal{L}$, let $z_\ell \in \{0,1\}$ be its selection variable, where $z_\ell = 1$ indicates that $\ell$ is chosen as the service instance assigned to the corresponding task in the final solution. Then, we formulate $\mathbf{P}$ as follows.
\begin{subequations}
    \allowdisplaybreaks
    \begin{align}
        (\mathbf{P}) \ \ \max \textstyle \sum_{\ell \in \mathcal{L}} & u(\ell) \cdot z_\ell \hspace{1cm} \label{eq:ilp} \\
        \text{subject to:} \hspace{1cm} \textstyle \sum_{\ell \in \mathcal{L}_j^\text{off}} b_\ell \cdot z_\ell \le B_j,& \ \forall m_j \in \mathcal{M} \label{eq:ilpa} \\
        \textstyle \sum_{\ell \in \mathcal{L}_k^\text{pro}} c_\ell \cdot z_\ell \le C_k,& \ \forall m_k \in \mathcal{M} \label{eq:ilpb} \\
        \textstyle \sum_{\ell \in \mathcal{L}_i^\text{task}} z_\ell \le 1,& \ \forall n_i \in \mathcal{N} \label{eq:ilpc} \\
        z_\ell \in \{0,1\},& \ \forall \ell \in \mathcal{L} \label{eq:ilpd}
    \end{align}
\end{subequations}
Eqs.~\eqref{eq:ilpa} and \eqref{eq:ilpb} are the ES resource constraints, while Eq.~\eqref{eq:ilpc} ensures that at most one service instance is selected for each task. The offloading constraints, processing constraints, and \textit{task deadline constraints} are implicitly enforced through the definition of service instances (Definition~\ref{def:instance}); thus, only ES resource constraints need to be explicitly enforced in the ILP formulation. As a result, an offloaded task contributes to the total accuracy improvement only if its deadline is met.

\textbf{Problem Complexity.} Given a solution to $\mathbf{P}$, we can efficiently verify whether it satisfies all the constraints in polynomial time, which indicates that the decision version of $\mathbf{P}$ belongs to the NP class. Furthermore, a well-known NP-Hard problem, the maximum weight three-dimensional matching problem \cite{burkard2009assignment}, can be reduced to a special case of $\mathbf{P}$, where each task has only one offloading mode and one candidate model, and each ES allocates its entire bandwidth to the tasks offloaded to it and all its computational resources to the tasks processed on it. Since this simplified instance of $\mathbf{P}$ is NP-Hard, the original problem $\mathbf{P}$ is also NP-Hard.

\section{Fast Separate Assignment Algorithm}\label{sec:algorithm}
In this section, we present the Fast Separate Assignment Algorithm ($\mathtt{FastSA}$) for problem $\mathbf{P}$ (Eqs. \eqref{eq:ilp}$\sim$\eqref{eq:ilpd}). Due to task forwarding, a service instance may compete with different sets of service instances for bandwidth and computational resources, which complicates the enforcement of ES resource constraints. $\mathtt{FastSA}$ addresses this challenge by exploiting the structure of service instances and separating them into two disjoint groups, which are scheduled independently to simplify resource contention. The algorithm then returns the schedule with the higher total accuracy improvement.

We then show that $\mathtt{FastSA}$ achieves a $\frac{1}{10}$-approximation guarantee for $\mathbf{P}$ in Subsection \ref{subsec:fastSA}, and present a linear-time implementation with respect to the size of $\mathcal{L}$ (i.e., $\mathcal{O}(|\mathcal{L}|)$) in Subsection \ref{subsec:fastSA-linear}. To the best of our knowledge, $\mathtt{FastSA}$ is the first constant-factor approximation algorithm for $\mathbf{P}$. 

\subsection{Fast Separate Assignment Algorithm (FastSA)}\label{subsec:fastSA}
Based on the allocated resource, we divide all service instances in $\mathcal{L}$ into four categories, LL, LH, HL, and HH, and denote the set of service instances in $\mathcal{L}$ that belong to each category as $\mathcal{L}_\text{LL}$, $\mathcal{L}_\text{LH}$, $\mathcal{L}_\text{HL}$, and $\mathcal{L}_\text{HH}$, respectively.
For a service instance $\ell \triangleq \left< n_i, m_j, m_k, b_\ell, c_\ell, \psi, a \right>$,
\begin{itemize}
    \item $\ell$ is in category LL if $b_{\ell} \le 1/2 \cdot B_j$ and $c_{\ell} \le  1/2 \cdot C_k$;
    \item $\ell$ is in category LH if $b_{\ell} \le 1/2 \cdot B_j$ and $c_{\ell} > 1/2 \cdot C_k$;
    \item $\ell$ is in category HL if $b_{\ell} > 1/2 \cdot B_j$ and $c_{\ell} \le 1/2 \cdot C_k$;
    \item $\ell$ is in category HH if $b_{\ell} > 1/2 \cdot B_j$ and $c_{\ell} > 1/2 \cdot C_k$.
\end{itemize}

$\mathtt{FastSA}$ (Algorithm~\ref{alg:fastSA}) divides $\mathcal{L}$ into two sets: $\mathcal{L}_\text{LL} \cup \mathcal{L}_\text{HL} \cup \mathcal{L}_\text{HH}$ and $\mathcal{L}_\text{LH}$. Let $\mathbb{L}_1$ be the vector of $\mathcal{L}_\text{LL} \cup \mathcal{L}_\text{HL} \cup \mathcal{L}_\text{HH}$ that is sorted in the order of LL, HL, and HH categories. Let $\mathbb{L}_2$ be the vector of $\mathcal{L}_\text{LH}$. $\mathtt{FastSA}$ first calls function $\mathtt{Assign}(\mathbf{u}, \mathbb{L}_1)$ to obtain a scheduling solution $\mathcal{S}_1$ for the combined set $\mathcal{L}_\text{LL} \cup \mathcal{L}_\text{HL} \cup \mathcal{L}_\text{HH}$ (line $2$), where $\mathbf{u}$ denotes the vector of $u(\ell)$ for all $\ell \in \mathcal{L}$. It then invokes $\mathtt{Assign}(\mathbf{u}, \mathbb{L}_2)$ to derive a scheduling solution $\mathcal{S}_2$ for the remaining set $\mathcal{L}_\text{LH}$ (line $3$). Finally, $\mathtt{FastSA}$ outputs the solution with the higher total accuracy improvement.
$\mathtt{Assign}$ is a recursive function (invoked recursively in line~$10$). At each recursion level, $\mathtt{FastSA}$ updates the accuracy values of the remaining service instances (line~$9$). For clarity, we denote the \textit{updated accuracy value} of service instance $\ell$ in the $r$-th level as $w^r(\ell)$, while $u(\ell)$ represents the \textit{original accuracy value}. Let $\mathbf{w}^r$ denote the vector of $w^r(\ell)$. For any service instance set $\mathcal{S}$, we define $w^r(\mathcal{S}) = \sum_{\ell \in \mathcal{S}} w^r(\ell)$ and $u(\mathcal{S}) = \sum_{\ell \in \mathcal{S}} u(\ell)$.

In line~$9$, the accuracy value $w^r(\ell)$ for each remaining $\ell$ is decomposed into $w_1^r(\ell)$ and $w_2^r(\ell)$ based on Eq.~\eqref{eq:fastSA1} and the chosen $\ell^r$ in line $8$. The first four conditions in Eq.~\eqref{eq:fastSA1} correspond to service instances that (i) belong to the same task as $\ell^r$, (ii) are offloaded to and processed on the same ESs as $\ell^r$, (iii) are only offloaded to the same ES as $\ell^r$, and (iv) are only processed on the same ES as $\ell^r$. The intuition is \textit{selecting $\ell^r$ reduces the opportunity for other service instances to be selected}. Specifically, $w_1^r(\ell)$ estimates the impact on $\ell$ if $\ell^r$ is selected, whereas $w_2^r(\ell)$ measures the marginal gain with respect to considered service instances in previous recursive levels. $\mathbf{w}_2^r$ is then used as the input accuracy vector $\mathbf{w}^{r+1}$ for the next recursive level (line~$10$). Based on the solution $\mathcal{S}^{r+1}$ returned from the $(r+1)$-th level, $\ell^r$ is included in the solution $\mathcal{S}^{r}$ if it satisfies the feasibility constraints \eqref{eq:ilpa}--\eqref{eq:ilpc} (line~$11$).

In the following, we first show that $\mathtt{FastSA}$ achieves a $\frac{1}{5}$-approximation for $\mathbf{P}$ when only the set $\mathcal{L}_\text{LL} \cup \mathcal{L}_\text{HL} \cup \mathcal{L}_\text{HH}$ is considered. We then prove that $\mathtt{FastSA}$ also achieves a $\frac{1}{5}$-approximation for $\mathbf{P}$ when considering only $\mathcal{L}_\text{LH}$. Finally, we show that $\mathtt{FastSA}$ is a $\frac{1}{10}$-approximation algorithm for $\mathbf{P}$.

\begin{algorithm}[t!]
    \caption{Fast Separate Assignment ($\mathtt{FastSA}$)}\label{alg:fastSA}

    \SetAlgoHangIndent{0pt} 
    \SetKwFunction{Assign}{$\mathtt{Assign}$}
    \SetKwProg{Fn}{}{:}{}
    \SetKwInOut{Input}{input}
    \SetKwInOut{Output}{output}
    \SetKw{Return}{return}
    \SetKw{Or}{or}
    \SetNoFillComment
    \SetAlgoNoEnd

    Let $\mathbb{L}_1$ be the vector of $\mathcal{L}_\text{LL} \cup \mathcal{L}_\text{HL} \cup \mathcal{L}_\text{HH}$ and is sorted in the order of LL, HL, and HH categories.
    Let $\mathbb{L}_2$ be the vector of $\mathcal{L}_\text{LH}$\;
    $\mathcal{S}_1 \leftarrow$ \Assign{$\mathbf{u}, \mathbb{L}_1$}\;
    $\mathcal{S}_2 \leftarrow$ \Assign{$\mathbf{u}, \mathbb{L}_2$}\;
    \leIf{$u(\mathcal{S}_1) \ge u(\mathcal{S}_2)$}{\Return $\mathcal{S}_1$}{\Return $\mathcal{S}_2$}

    \BlankLine
    \Fn{\Assign{$\mathbf{w}^r, \mathbb{F}$}}
    {
        Remove all instances $\ell$ from $\mathbb{F}$ with $w^r(\ell) \le 0$\;
        \lIf{$\mathbb{F} = \emptyset$}{\Return $\mathcal{S}^r=\emptyset$}
        Suppose $\ell^r \triangleq \left< n_i, m_j, m_k, b_{\ell^r}, c_{\ell^r}, \psi, a \right>$ is the leftmost service instance in $\mathbb{F}$\;
        \scalebox{1}{Decompose $\mathbf{w}^r$ into $\mathbf{w}^r_1$ and $\mathbf{w}^r_2$ (i.e., $\mathbf{w}^r$=$\mathbf{w}^r_1$+$\mathbf{w}^r_2$)} such that $\forall \ell \in \mathbb{F}, w_1^r(\ell) =$
        \begin{equation}\label{eq:fastSA1}
          w^r(\ell^r) \cdot
        \begin{cases}
            1 & \text{if } \ell \in \mathcal{L}_i^\text{task} \\
            2 (\frac{b_\ell}{B_j} + \frac{c_\ell}{C_k}) & \text{if }  \ell \in \mathcal{L}_j^\text{off} \cap \mathcal{L}_k^\text{pro} \setminus  \mathcal{L}_i^\text{task} \\
            2 \frac{b_\ell}{B_j} & \text{if }  \ell \in \mathcal{L}_j^\text{off} \setminus  (\mathcal{L}_k^\text{pro} \cup \mathcal{L}_i^\text{task}) \\
            2 \frac{c_\ell}{C_k} & \text{if }  \ell \in \mathcal{L}_k^\text{pro} \setminus (\mathcal{L}_j^\text{off} \cup \mathcal{L}_i^\text{task}) \\
            0 & \text{otherwise}
        \end{cases}
    \end{equation}
    
    $\mathcal{S}^{r+1} \leftarrow $ \Assign{$\mathbf{w}_2^r$, $\mathbb{F}$}\;
    \leIf{\scalebox{1}{$\mathcal{S}^{r+1} \cup \{\ell^r\}$ is feasible}}{\Return \scalebox{1}{$\mathcal{S}^r \leftarrow \mathcal{S}^{r+1} \cup \{\ell^r\}$}}{\Return \scalebox{1}{$\mathcal{S}^r \leftarrow \mathcal{S}^{r+1}$}}
    }
\end{algorithm}

\begin{theorem}\label{theorem:fastSA1}
    Let $\mathcal{Q}_1$ be an optimal solution of $\mathbf{P}$ when only considering the set $\mathcal{L}_\text{LL} \cup \mathcal{L}_\text{HL} \cup \mathcal{L}_\text{HH}$. Let $\mathcal{S}_1$ be the solution returned in line $2$ of $\mathtt{FastSA}$. Then, we have $u(\mathcal{S}_1) \ge \frac{1}{5}u(\mathcal{Q}_1)$.
\end{theorem}
\begin{proof}
    Let the deepest recursive level of $\mathtt{Assign}$ be level $\delta$ (i.e., when $\mathbb{F}=\emptyset$ in line $7$), and the initial call be level $1$. We use simple induction to prove that $w^r(\mathcal{S}^r) \ge \frac{1}{5}w^r(\mathcal{Q}_1)$ for $r=\delta, \delta-1, ..., 1$. Since $\mathbf{w}^1=\mathbf{u}$ and $\mathcal{S}^1$ is the solution returned in line $2$ of $\mathtt{FastSA}$, this theorem can thus be proved.

    \textbf{Base Case} ($r=\delta$): When $r=\delta$, $\mathbb{F}=\emptyset$ after removing all service instances with non-positive accuracy value (line $6$), implying that $w^r(\ell) \le 0, \forall \ell \in \mathbb{F}$, before line $6$ is executed. The returned $\mathcal{S}^r=\emptyset$, so $w^r(\mathcal{S}^r) \ge w^r(\mathcal{Q}_1) \ge \frac{1}{5}w^r(\mathcal{Q}_1)$.

    \textbf{Inductive Step} ($r < \delta$): When $r < \delta$, we show that if $w^{r+1}(\mathcal{S}^{r+1}) \ge \frac{1}{5}w^{r+1}(\mathcal{Q}_1)$ holds, then $w^{r}(\mathcal{S}^r) \ge \frac{1}{5}w^r(\mathcal{Q}_1)$. Here, $\mathcal{S}^{r+1}$ is the solution returned in line $10$ from the lower level. Suppose $w^{r+1}(\mathcal{S}^{r+1}) \ge \frac{1}{5}w^{r+1}(\mathcal{Q}_1)$. Since we assign $\mathbf{w}_2^r$ as the accuracy vector for the lower level in line $10$, $\mathbf{w}_2^r = \mathbf{w}^{r+1}$. Then, we have $w_2^r(\mathcal{S}^{r+1}) \ge \frac{1}{5}w_2^r(\mathcal{Q}_1)$. Let $\ell^r \triangleq \left< n_i, m_j, m_k, b_{\ell^r}, c_{\ell^r}, \psi, a \right>$ (line $8$). Based on the first condition in Eq. \eqref{eq:fastSA1}, $w_1^r(\ell^r)=w^r(\ell^r)$, so $w_2^r(\ell^r)=0$; thus,
    \begin{equation}\label{eq:fastSA2}
        \textstyle w_2^r(\mathcal{S}^r) = w_2^r(\mathcal{S}^{r+1}) \ge \frac{1}{5}w_2^r(\mathcal{Q}_1).
    \end{equation}
    If $\ell^r$ can be added to $\mathcal{S}^r$ in line $11$, \scalebox{0.9}{$w_1^r(S^r) \ge w_1^r(\ell^r)=w^r(\ell^r)$.}\\
    If $\ell^r$ cannot be added to $\mathcal{S}^r$ in line $11$ due to constraint \eqref{eq:ilpc} (i.e., at most one instance of the same task can be selected), one $\ell \in \mathcal{L}_i^\text{task}$ must already exist in $\mathcal{S}^{r+1}$. Based on the first condition in Eq. \eqref{eq:fastSA1}, $w_1^r(S^r) = w_1^r(S^{r+1}) \ge w^r(\ell^r)$.\\
    If $\ell^r$ cannot be added to $\mathcal{S}^r$ in line $11$ due to the ES resource constraints \eqref{eq:ilpa} or \eqref{eq:ilpb}, we proceed with the analysis based on the category of $\ell^r$. Given the sorted vector $\mathbb{L}_1$ and the ``leftmost'' selection strategy in line~$8$, the service instances in line~$8$ are chosen in the order of LL, HL, and HH categories. Based on lines~$10$--$11$, \textbf{adding $\ell^r$ into the solution $\mathcal{S}^r$ occurs in the reverse category order}, i.e., HH, HL, and LL. Thus, we consider the following three cases:
    \begin{itemize}
        \item Case 1: $\ell^r \in \mathcal{L}_\text{HH}$. In this case, $\mathcal{S}^{r+1}$ is a subset of $\mathcal{L}_\text{HH}$. Since $\ell^r$ cannot be added into $\mathcal{S}^{r+1}$ due to ES resource constraint, there must exist one $\ell \in \mathcal{S}^{r+1}$ in which the task is offloaded to ES $m_j$ or processed on ES $m_k$ (maybe both). By definition, $\frac{b_\ell}{B_j} > \frac{1}{2}$ and $\frac{c_\ell}{C_k} > \frac{1}{2}$ for all $\ell \in \mathcal{L}_\text{HH}$. Based on the second, third, or fourth condition in Eq. \eqref{eq:fastSA1}, we have $w_1^r(S^r) = w_1^r(S^{r+1}) \ge w^r(\ell^r)$.
        \item Case 2: $\ell^r \in \mathcal{L}_\text{HL}$. In this case, $\mathcal{S}^{r+1}$ is a subset of $\mathcal{L}_\text{HL} \cup \mathcal{L}_\text{HH}$. By definition, $\frac{b_\ell}{B_j} > \frac{1}{2}$ and $\frac{c_\ell}{C_k} \le \frac{1}{2}$ for all $\ell \in \mathcal{L}_\text{HL}$. If $\ell^r$ cannot be added into $\mathcal{S}^r$ due to the ES bandwidth constraint, there must exist one service instance $\ell \in \mathcal{S}^{r+1}$ in which the task is offloaded to ES $m_j$. As $\frac{b_\ell}{B_j} > \frac{1}{2}$ for all $\ell \in \mathcal{L}_\text{HL} \cup \mathcal{L}_\text{HH}$, based on the second or third condition in Eq. \eqref{eq:fastSA1}, $w_1^r(S^r) = w_1^r(S^{r+1})\ge w^r(\ell^r)$. If $\ell^r$ cannot be added into $\mathcal{S}^{r+1}$ due to the ES computational resource constraint, at least $\frac{1}{2}C_k$ amount of ES $m_k$'s computational resource have been allocated to service instances in $\mathcal{S}^{r+1}$. Based on the second or fourth condition in Eq. \eqref{eq:fastSA1}, $w_1^r(S^r) = w_1^r(S^{r+1}) \ge w^r(\ell^r)$.
        \item Case 3: $\ell^r \in \mathcal{L}_\text{LL}$. Since $\frac{b_\ell}{B_j} \le \frac{1}{2}$ and $\frac{c_\ell}{C_k} \le \frac{1}{2}$ for all $\ell \in \mathcal{L}_\text{LL}$, if $\ell^r$ cannot be added into $\mathcal{S}^r$ due to ES resource constraints, at least $\frac{1}{2}B_j$ amount of ES $m_j$'s bandwidth or $\frac{1}{2}C_k$ amount of ES $m_k$'s computational resource have been allocated to service instances in $\mathcal{S}^{r+1}$ (maybe both); based on the second to forth conditions in Eq. \eqref{eq:fastSA1}, we have $w_1^r(S^r) = w_1^r(S^{r+1}) \ge w^r(\ell^r)$.
    \end{itemize}
    
    In summary, regardless of whether $\ell^r$ can be added to $\mathcal{S}^{r+1}$ in line $11$, $w_1^r(S^r) \ge w^r(\ell^r)$. Meanwhile, based on Eq. \eqref{eq:fastSA1}, the value $w_1^r(\mathcal{Q}_1)$ can be \textit{maximized} if the optimal solution $\mathcal{Q}_1$ contains one service instance in $\mathcal{L}_i^\text{task} \setminus (\mathcal{L}_j^\text{off} \cup \mathcal{L}_k^\text{pro})$, and all bandwidth of ES $m_j$ (i.e., $B_j$) and all computational resource of ES $m_k$ (i.e., $C_k$) are allocated to service instances in $(\mathcal{L}_j^\text{off} \cup \mathcal{L}_k^\text{pro}) \setminus \mathcal{L}_i^\text{task}$ (only these service instances have non-zero $w_1^r(\ell)$); in this case, $w_1^r(\mathcal{Q}_1) \le w^r(\ell^r) + 4w^r(\ell^r) = 5w^r(\ell^r)$. Hence,
    \begin{equation}\label{eq:fastSA3}
        \textstyle w_1^r(S^r) \ge \frac{1}{5} \cdot w_1^r(\mathcal{Q}_1).
    \end{equation}
    Given Eq. \eqref{eq:fastSA2}, Eq. \eqref{eq:fastSA3}, and $\mathbf{w}^r = \mathbf{w}_1^r + \mathbf{w}_2^r$, we have $w^r(S^r) \ge \frac{1}{5}w^r(\mathcal{Q}_1)$.
    Combining the results in the base case and the inductive step, we have $w^r(S^r) \ge \frac{1}{5}w^r(\mathcal{Q}_1)$ for $r=\delta,...,1$ based on simple induction.
\end{proof}

\begin{theorem}\label{theorem:fastSA2}
    Let $\mathcal{Q}_2$ be an optimal solution of $\mathbf{P}$ when only considering the set $\mathcal{L}_\text{LH}$, and $\mathcal{S}_2$ be the solution returned in line $3$ of $\mathtt{FastSA}$. Then, $u(\mathcal{S}_2) \ge \frac{1}{5}u(\mathcal{Q}_2)$.
\end{theorem}
\begin{proof}
    The proof of this theorem follows a similar structure to that of Theorem~\ref{theorem:fastSA1}. The key difference is that when $\ell^r$ cannot be added to $\mathcal{S}^r$ in line~$11$ due to constraints \eqref{eq:ilpa} or \eqref{eq:ilpb} (i.e., the ES resource limitations), only one case (i.e., analogous to the Case~2 analysis in Theorem~\ref{theorem:fastSA1}) needs to be considered. The complete proof is omitted here.
\end{proof}

Combining the results from both Theorem \ref{theorem:fastSA1} and Theorem \ref{theorem:fastSA2}, we have the following conclusion.

\begin{theorem}\label{theorem:fastSA3}
    $\mathtt{FastSA}$ is a $\frac{1}{10}$-approximation guarantee for $\mathbf{P}$.
\end{theorem}
\begin{proof}
    Let $\mathcal{S}$ be the solution returned by $\mathtt{FastSA}$, $\mathcal{S}_1$ be the solution returned in line $2$ of $\mathtt{FastSA}$, and $\mathcal{S}_2$ be the solution returned in line $3$ of $\mathtt{FastSA}$. Let $\mathcal{Q}$ be the optimal solution of $\mathbf{P}$, $\mathcal{Q}_1$ be the optimal solution of $\mathbf{P}$ when only considering service instances in $\mathcal{L}_\text{LL} \cup \mathcal{L}_\text{HL} \cup \mathcal{L}_\text{HH}$, and $\mathcal{Q}_2$ be the optimal solution of $\mathbf{P}$ when only considering service instances in $\mathcal{L}_\text{LH}$.

    Any solution of $\mathbf{P}$ can be divided into a subset of $\mathcal{L}_\text{LL} \cup \mathcal{L}_\text{HL} \cup \mathcal{L}_\text{HH}$ and a subset of $\mathcal{L}_\text{LH}$. Therefore, $u(\mathcal{Q}) \le u(\mathcal{Q}_1) + u(\mathcal{Q}_2)$. Besides, $u(\mathcal{S}) = \max\{u(\mathcal{S}_1), u(\mathcal{S}_2)\}$. Based on Theorem \ref{theorem:fastSA1} and Theorem \ref{theorem:fastSA2}, we have $u(\mathcal{Q}_1) \le 5u(\mathcal{S}_1)$ and $u(\mathcal{Q}_2) \le 5u(\mathcal{S}_2)$. Then, we have $u(\mathcal{Q}) \le u(\mathcal{Q}_1) + u(\mathcal{Q}_2) \le 5u(\mathcal{S}_1) + 5u(\mathcal{S}_2) \le 10\cdot \max\{u(\mathcal{S}_1), u(\mathcal{S}_2)\} = 10u(\mathcal{S})$.
\end{proof}

In $\mathtt{FastSA}$, the function $\mathtt{Assign}$ has at most $|\mathcal{L}|$ recursive levels (i.e., the size of $\mathcal{L}$), and the accuracy values of at most $|\mathcal{L}|$ service instances are updated in each recursive level. Thus, the time complexity of $\mathtt{FastSA}$ is $\mathcal{O}(|\mathcal{L}|^2)$. In Subsection~\ref{subsec:fastSA-linear}, we show how to implement $\mathtt{FastSA}$ with a linear runtime complexity with respect to the size of $\mathcal{L}$, i.e., $\mathcal{O}(|\mathcal{L}|)$.

\subsection{Linear-Time Implementation for FastSA}\label{subsec:fastSA-linear}
Let $\ell^r \triangleq \left< n_i, m_j, m_k, b_{\ell^r}, c_{\ell^r}, \psi, a \right>$ denote the service instance considered in the $r$-th recursive level of the function $\mathtt{Assign}$ (line~$8$). According to Eq.~\eqref{eq:fastSA1}, the accuracy value of $\ell^r$ is reduced only when service instances in $\mathcal{L}_i^\text{task} \cup \mathcal{L}_j^\text{off} \cup \mathcal{L}_k^\text{pro}$ are processed in earlier recursive levels $s < r$. Let $\tilde{u}(\ell^r)$ denote the total accuracy value \textit{reduction} for $\ell^r$ when it is considered, i.e., $w^r(\ell^r) = u(\ell^r) - \tilde{u}(\ell^r)$. Based on the decomposition rule in Eq.~\eqref{eq:fastSA1}, $\tilde{u}(\ell^r)$ can be expressed as
\begin{equation}\label{eq:fastSA-linear1}
    \textstyle \tilde{u}(\ell^r) = U_i^\text{task} + 2 \cdot U_j^\text{off}\cdot \frac{b_{\ell^r}}{B_j} + 2 \cdot U_k^\text{pro}\cdot \frac{c_{\ell^r}}{C_k},
\end{equation}
where $U_i^\text{task}$, $U_j^\text{off}$, and $U_k^\text{pro}$ represent the reduced accuracy caused by service instances processed in earlier recursive levels that belong to the same task $n_i$ (first condition in Eq.~\eqref{eq:fastSA1}), are offloaded to the same ES $m_j$ (second or third condition in Eq.~\eqref{eq:fastSA1}), and are processed on the same ES $m_k$ (second or fourth condition in Eq.~\eqref{eq:fastSA1}), respectively. Specifically, $U_i^\text{task} = \sum_{\substack{s=1,  \ell^s \in \mathcal{L}_i^\text{task}}}^{r-1}w^s(\ell^s)$, $U_j^\text{off} = \sum_{\substack{s=1, \ell^s \in \mathcal{L}_j^\text{off} \setminus \mathcal{L}_i^\text{task}}}^{r-1}w^s(\ell^s)$, and $U_k^\text{pro} = \sum_{\substack{s=1, \ell^s \in \mathcal{L}_k^\text{pro} \setminus \mathcal{L}_i^\text{task}}}^{r-1} w^s(\ell^s)$.

Motivated by these observations, we simplify $\mathtt{FastSA}$ as shown in Algorithm~\ref{alg:fastSA-linear}. We introduce five auxiliary vectors, $\mathbf{U}_n$, $\mathbf{U}_m^\text{off}$, $ \mathbf{U}_{nm}^\text{off}$, $\mathbf{U}_m^\text{pro}$, and $\mathbf{U}_{nm}^\text{pro}$, to store accumulated accuracy reductions. Specifically, $\mathbf{U}_n[i]$, $\mathbf{U}_m^\text{off}[j]$, $\mathbf{U}_{nm}^\text{off}[i,j]$, $\mathbf{U}_m^\text{pro}[k]$, and $\mathbf{U}_{nm}^\text{pro}[i,k]$ store the accuracy decrements for service instances in $\mathcal{L}_i^\text{task}$, $\mathcal{L}_j^\text{off}$, $\mathcal{L}_i^\text{task}\cap \mathcal{L}_j^\text{off}$, $\mathcal{L}_k^\text{pro}$, and $\mathcal{L}_i^\text{task}\cap \mathcal{L}_k^\text{pro}$, respectively. In addition, the vector $\mathbb{L}_{candi}$ stores candidate service instances, and the set $\mathcal{S}$ stores the final solution.

\begin{algorithm}[t!]
    \caption{Linear-Time Implementation for $\mathtt{FastSA}$}\label{alg:fastSA-linear}

    \SetAlgoHangIndent{0pt} 
    \SetKwFunction{Assign}{$\mathtt{Assign}$}
    \SetKwProg{Fn}{}{:}{}
    \SetKwInOut{Input}{input}
    \SetKwInOut{Output}{output}
    \SetKw{Return}{return}
    \SetKw{Or}{or}
    \SetNoFillComment
    \SetAlgoNoEnd

    \tcc{\small \textsf{
    consider categories LL, HL, and HH in order
    }}
    $\mathcal{S}_1 \leftarrow$ \Assign{$[\text{LL}, \text{HL}, \text{HH}]$}\;
    $\mathcal{S}_2 \leftarrow$ \Assign{$[\text{LH}]$}\tcp*{\small \textsf{category LH only}}
    \leIf{$u(\mathcal{S}_1) \ge u(\mathcal{S}_2)$}{\Return $\mathcal{S}_1$}{\Return $\mathcal{S}_2$}

    \BlankLine
    \Fn{\Assign{CategoryList}}
    {
        \scalebox{0.92}{Initialize vectors $\mathbf{U}_n, \mathbf{U}_m^\text{off}, \mathbf{U}_{nm}^\text{off}, \mathbf{U}_m^\text{pro}, \mathbf{U}_{nm}^\text{pro}, \mathbb{L}_{candi}, \mathcal{S}$}\;
        \For{$category \in \text{CategoryList}$}{
            \For(\small \hspace{0.5cm} \textsf{// $\ell^r$: $r$-th element in $\mathcal{L}$}){$r \leftarrow 1$ \KwTo $|\mathcal{L}|$}
            {
                \lIf{$\ell^r$ is not in $category$}{\textbf{continue}}
                
                Suppose $\ell^r \triangleq \left< n_i, m_j, m_k, b_{\ell^r}, c_{\ell^r}, \psi, a_r \right>$\;
                Let \scalebox{1}{$U_i^\text{task} = \mathbf{U}_n[i], U_j^\text{off} = \mathbf{U}_m^\text{off}[j] - \mathbf{U}_{nm}^\text{off}[i,j]$,}
                and $ U_k^\text{pro} = \mathbf{U}_m^\text{pro}[k] - \mathbf{U}_{nm}^\text{pro}[i,k]$\;

                $\tilde{u}(\ell^r) = U_i^\text{task}$$+$$2 U_j^\text{off} \cdot \frac{b_{\ell^r}}{B_k}$$+$$2U_k^\text{pro} \cdot \frac{c_{\ell^r}}{C_k}$ (Eq.~\eqref{eq:fastSA-linear1})\;
                $w^r(\ell^r) = u(\ell^r) - \tilde{u}(\ell^r)$\;
                
                \lIf{$w^r(\ell^r) <= 0$}{\textbf{continue}}
                    
                Append $\ell^r$ into the end of vector $\mathbb{L}_{candi}$\;
                \scalebox{1}{$\mathbf{U}_n[i] $=$ \mathbf{U}_n[i] $+$ w^r(\ell^r)$,
                $\mathbf{U}_m^\text{off}[j] $=$ \mathbf{U}_m^\text{off}[j] $+$w^r(\ell^r)$,}
                \scalebox{0.93}{$\mathbf{U}_{nm}^\text{off}[i,j] $=$ \mathbf{U}_{nm}^\text{off}[i,j] $+$w^r(\ell^r)$,
                $\mathbf{U}_m^\text{pro}[k] $=$ \mathbf{U}_m^\text{pro}[k] $+$w^r(\ell^r)$,}
                $\mathbf{U}_{nm}^\text{pro}[i,k] $=$ \mathbf{U}_{nm}^\text{pro}[i,k] $+$w^r(\ell^r)$\;
            }
        }

        \For(\hspace{0.3cm} \scalebox{0.8}{\textsf{// $\ell^r$: $r$-th element in $\mathbb{L}_{candi}$}}){$r \leftarrow |\mathbb{L}_{candi}|$ \KwTo $1$}
        {
            \lIf{$\mathcal{S} \cup \{\ell^r\}$ is feasible}
            {
                $\mathcal{S} \leftarrow \mathcal{S} \cup \{\ell^r\}$
            }
        }
        
        \Return $\mathcal{S}$\;
    }
\end{algorithm}

Algorithm~\ref{alg:fastSA-linear} first calls the function $\mathtt{Assign}\text{([LL, HL, HH])}$ to handle service instances in categories LL, HL, and HH sequentially (line~$1$). For each category, $\mathtt{Assign}$ traverses the entire $\mathcal{L}$ and only considers the service instances belonging to the current category (lines~$6$--$8$). For each considered service instance $\ell^r$, the updated accuracy value $w^r(\ell^r)$ is computed according to Eq.~\eqref{eq:fastSA-linear1} (lines~$10$--$12$). If $w^r(\ell^r)$ is positive, $\ell^r$ is appended to $\mathbb{L}_{candi}$, and the auxiliary vectors are updated accordingly (lines~$13$--$15$). Then, candidates in $\mathbb{L}_{candi}$ are added to the solution $\mathcal{S}$ in reverse order. Algorithm~\ref{alg:fastSA-linear} subsequently calls $\mathtt{Assign}\text{([LH])}$ to schedule service instances in category LH (line~$2$). Finally, the solution achieving the higher total accuracy improvement is selected as the output of $\mathtt{FastSA}$.

The sequence of candidate service instances in $\mathbb{L}_{candi}$ generated by the $\mathtt{Assign}$ function in Algorithm~\ref{alg:fastSA-linear} exactly mirrors the sequence produced by the recursive $\mathtt{Assign}$ function in $\mathtt{FastSA}$. Moreover, the accuracy reduction applied to each instance remains identical. Therefore, \textbf{the theoretical guarantee established in Theorem~\ref{theorem:fastSA3} still holds.} Algorithm~\ref{alg:fastSA-linear} performs four traversals over $\mathcal{L}$, during which the accuracy value of each service instance is only updated once. Thus, the time complexity of Algorithm~\ref{alg:fastSA-linear} is $\mathcal{O}(|\mathcal{L}|) = \mathcal{O}(2 N M^2 A \cdot \min\{B, C\})$.

\section{Fault-Tolerant Offloading Framework} \label{sec:framework}
This section presents a fault-tolerant task offloading framework for real-time tasks. We first introduce the fault detection and recovery mechanism to handle both backhaul link and ES failures, and then present the task offloading workflow.

\subsection{Link-State Synchronization and Global Scheduler Election}\label{subsec:framework-bn} 

To maintain global network awareness, we adopt a flooding-based synchronization protocol (similar to OSPF~\cite{ospfv2moy}) to disseminate backhaul connectivity information among all ESs.

Each ES is associated with a gNB, a server, a scheduler, and a router. The gNB handles data offloading, and the server hosts services for EDs. The scheduler remains idle unless it is elected as the global scheduler, which then determines service placement and resource allocation for ED requests. The router is responsible for backhaul network synchronization. It maintains a Link State Advertisement (LSA) database that stores both the LSA generated by itself (referred to as the self-LSA) and LSAs received from other ESs (Fig.~\ref{fig:framework-lsa}). Each router records the \textit{link state} of its ES in the self-LSA, which includes the origin ES, a timestamp indicating the freshness of the link state, and a list of neighboring ESs it has detected. When a link state change is detected (e.g., a new neighbor is discovered or an existing one is disconnected), the router updates the timestamp and the list of neighboring ESs in the self-LSA, then broadcasts the updated self-LSA message to all directly connected ESs.

Upon receiving an LSA from another ES, the router compares the timestamp of the received LSA with that of the corresponding LSA (with the same origin ES) in its database. If the received LSA is newer, the router updates its database to maintain freshness and forwards the new LSA to all neighboring ESs \textit{except the one from which it was received}, as that ES already possesses the update; otherwise, the received LSA is discarded. This ensures that each updated LSA traverses every link in each direction at most once, thereby preventing redundant traffic and minimizing backhaul overhead. An example of this process is illustrated in Fig.~\ref{fig:framework-lsa}.


\begin{figure}
    \centering
    \includegraphics[width=0.95\linewidth, page=2]{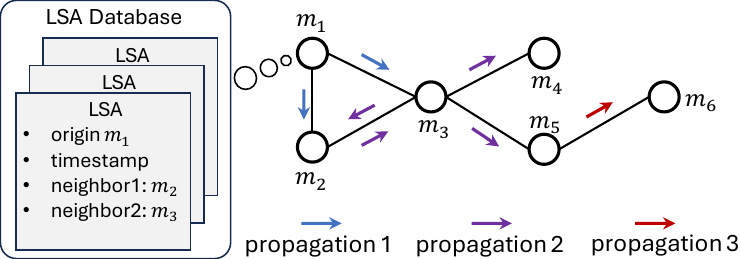}
    \caption{An example of ES $m_1$ disseminates its self-LSA to other ESs. When the self-LSA of $m_1$ is updated, $m_1$ broadcasts the self-LSA to its neighboring ESs $m_2$ and $m_3$ (first propagation). In the second propagation, $m_3$ forwards the LSA to $m_2$, $m_4$, and $m_5$, while $m_2$ also forwards the LSA to $m_3$. During the third propagation, $m_5$ further forwards the LSA to $m_6$.}
    \label{fig:framework-lsa}
\end{figure}

Through this flooding process, once a link-state change is detected, all ESs within the same backhaul network (e.g., in which every pair of ESs is connected by a path) can be informed within a bounded time, as shown in the following.
\begin{proposition}\label{prop:convergence}
Let $\Delta$ denote the longest shortest path between any two ESs in a backhaul network (i.e., the diameter). The updated self-LSA of any ES can be disseminated to all other ESs within at most $\Delta$ hops of LSA message propagation.
\end{proposition}

\textit{Global Scheduler Election.} Whenever an ES receives an LSA message, it initiates a global scheduler election timer using the timestamp contained in the message. If multiple LSAs are received from different ESs, the timer is set according to the largest timestamp among them, which represents the time when the most recent link-state change is detected. Since any LSA can be disseminated to all ESs within a bounded time (Proposition~\ref{prop:convergence}), an appropriate delay can be set for the timer based on the largest timestamp such that, upon expiration, all link-state updates are fully disseminated across the network.

Once the backhaul network is synchronized, each ES within the same backhaul holds an identical LSA database, which is used by Dijkstra's algorithm to determine the shortest paths among ESs in the backhaul network. A deterministic and self-consistent election procedure is then independently executed by all ESs in the same backhaul to select a common global scheduler without additional coordination. The election rule ensures a unique outcome within each backhaul: the ES connected to the largest number of neighboring ESs is selected as the global scheduler, and in case of a tie, the ES with the smallest identifier is chosen. This rule ensures consistency among all ESs in the same backhaul. When a backhaul network is partitioned into multiple disjoint backhaul networks due to backhaul link or ES failures, each backhaul independently elects its own global scheduler.

After a new global scheduler is elected, a periodic scheduling timer is initiated to repeatedly invoke the scheduling algorithm, which determines service deployment and resource allocation for pending ED requests by solving $\mathbf{P}$ based on the resource in its backhaul. When the backhaul network is partitioned, multiple global schedulers operate concurrently, each making decisions exclusively for its own backhaul.

\subsection{Failure Detection and Recovery}

Each ES connects neighboring ESs via Ethernet interfaces using the Point-to-Point Protocol (PPP). Each ES can have multiple interfaces (termed pppIf), each of which connects to at most one neighboring ES, while unused interfaces remain unconnected. At the system-design level, these point-to-point links represent dedicated backhaul connections between ESs. To support predictable end-to-end behavior, the wired backhaul is assumed to provide deterministic transmission with fixed data rates, as can be realized using TSN-capable Ethernet or similar technologies. Accordingly, backhaul delay is modeled as a deterministic function of data size and hop count, as defined in Section~\ref{sec:system-architecture}.

To monitor link availability, every ES \textit{periodically} broadcasts HELLO messages through all its pppIfs. Upon receiving a HELLO message, the neighboring ES replies with an acknowledgment (ACK). If the ACK is received within a predefined time interval, the corresponding link is marked as \textit{active}. Conversely, if no ACK is received for two consecutive HELLO message broadcasts (to avoid excessive sensitivity), the link is considered \textit{inactive}.

This framework considers both ES and backhaul link failures. When a link between two ESs fails, both ESs detect the connection loss. Upon confirmation after two consecutive missed ACKs, each ES updates its self-LSA. Similarly, if an ES (whether hosting the global scheduler or not) fails, all its neighboring ESs will detect the loss of connectivity and update their respective self-LSAs after the same confirmation interval. Then, each ES with updated-LSA broadcasts its self-LSA following the synchronization procedure in Subsection~\ref{subsec:framework-bn}.

Upon receiving an LSA, an ES (e.g., $m_k$) forwards the message following the procedure in Subsection~\ref{subsec:framework-bn} and updates its LSA database accordingly. In addition, ES $m_k$ terminates all services hosted on its server and sends termination grants to the tasks that have offloaded to it. The ES then releases its occupied bandwidth and computational resources, and reports its updated status to the newly elected global scheduler once it is ready. If ES $m_k$ previously hosts a global scheduler, the scheduler state is set to \textit{idle}. 

Furthermore, the discovery of a new neighboring ES automatically triggers a link-state update. Consequently, when a previously failed link or ES recovers, it is reintegrated into the existing topology. This failure detection and recovery mechanism provides strong robustness against both ES and backhaul link failures, ensuring continuous MEC service. 


\subsection{Task Offloading Workflow}\label{subsec:framework-offload}
A detailed task offloading workflow is illustrated in Fig.~\ref{fig:framework-offloading}. When an ED decides to offload a task (e.g., $n_i$), it first transmits a service request to its nearest ES (e.g., $m_j$) via the 5G wireless network, including its deadline $T_i$ and offload data sizes $\theta_i^1$, $\theta_i^2$, etc. Upon receiving the request, ES~$m_j$ temporarily stores it in a local buffer. Once the global scheduler is elected, all buffered requests are forwarded to it for scheduling.


The global scheduler is invoked periodically to continuously serve surrounding EDs. Each invocation defines a \textit{scheduling cycle}, and its duration is referred to as the \textit{scheduling interval}.

At the beginning of each scheduling cycle, every ED broadcasts an SRS to nearby gNBs. Upon receiving the SRS, each gNB selects an appropriate MCS to ensure reliable wireless transmission, thereby determining the achievable uplink data rate $\mu_{ij}$ between the ED and the ES (e.g., $m_j$) associated with that gNB. This uplink rate, together with the current abstracted resource availability at ES~$m_j$ (i.e., $B_j$ and $C_j$), is reported to the global scheduler.

Using the latest resource reports from all ESs, the global scheduler invokes the scheduling algorithm (e.g., $\mathtt{FastSA}$) to schedule pending requests by solving problem~$\mathbf{P}$ (defined in Section~\ref{sec:system-architecture}). Once a scheduling decision is obtained, an ES grant is issued to the server associated with each designated process-ES (e.g., $m_k$) to initialize the corresponding service. After successful service initialization, an offload grant is forwarded to the gNB associated with the designated offload-ES (e.g., $m_j$). The gNB interprets the offload grant to perform PHY/MAC-layer resource block scheduling and to notify the corresponding ED that task $n_i$ may begin offloading.

When a job instance (i.e., task input data) of $n_i$ is generated, it is first offloaded to the gNB associated with offload-ES~$m_j$, and then forwarded to the server associated with process-ES~$m_k$ via the backhaul network for processing. After processing completes, the result is returned to ES $m_j$ and then delivered to the ED. During offloading, the gNB continuously monitors the wireless channel quality. If channel conditions degrade such that timely offloading cannot be guaranteed, the gNB temporarily suspends the offloading process until channel quality recovers. During this suspension period, the task uses locally computed results to ensure safe operation; however, no utility (e.g., accuracy improvement) is accrued in this case.

\textit{The scheduling cycle repeats until a link-state change is detected in the backhaul network.} After the network is resynchronized and a new global scheduler is elected, a new scheduling cycle begins.

\textit{Implementation.} We implement the proposed framework on top of the open-source MEC simulator for real-time applications, \textit{mecRT}~\cite{gao2025real}, which provides a basic task offloading framework. However, the original \textit{mecRT} only supports scenarios where all ESs are directly connected to a fixed centralized global scheduler. It lacks data forwarding, backhaul network synchronization, and fault detection and recovery mechanisms. In this work, we extend \textit{mecRT} by integrating our fault-tolerant framework\footnote{The source code is available at \url{https://github.com/gaochuanchao/mecRT}.}, thereby enabling general and resilient MEC simulation scenarios.

\begin{figure}
    \centering
    \includegraphics[width=0.95\linewidth, page=3]{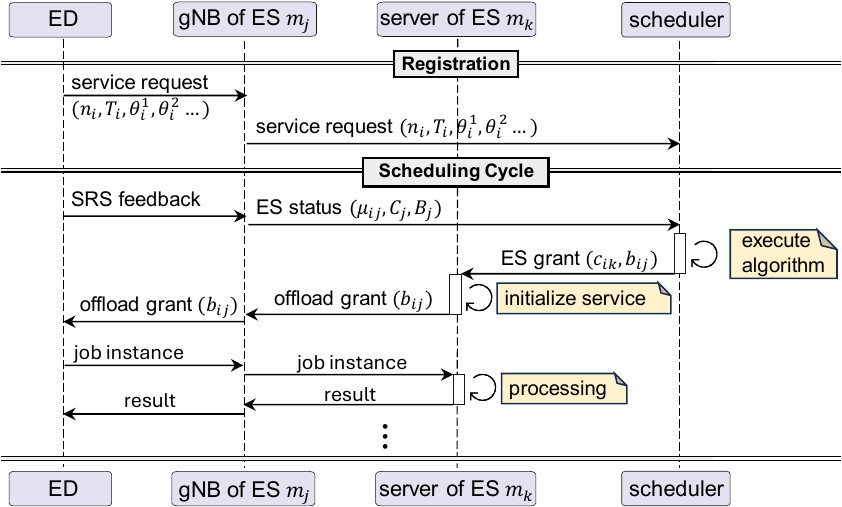}
    \caption{Task Offloading Workflow (task $n_i$, ESs $m_j, m_k$ are used as example).}
    \label{fig:framework-offloading}
\end{figure}

\section{Experimental Evaluation}\label{sec:experiment}
This section evaluates $\mathtt{FastSA}$ by comparing it with three state-of-the-art benchmark algorithms from the literature. We first examine $\mathtt{FastSA}$ in a fault-free MEC environment and then analyze the impact of backhaul link and ES failures on $\mathtt{FastSA}$. All experiments are conducted on a desktop PC with an Intel(R) Xeon(R) W-2235 CPU and 32 GB of RAM.

\subsection{Simulation Setup}
We evaluate $\mathtt{FastSA}$ using the simulator \textit{mecRT}~\cite{gao2025real}, which provides detailed 5G network modeling and task offloading control logic.
The simulation emulates an intelligent transportation scenario, where MEC services are provided to vehicles. Real taxi GPS trajectory data \cite{taxi-data} is used to model vehicle mobility. Specifically, we extract all taxi trajectories passing through a $1$ km $\times$ $1$ km urban area in Shanghai over a $15$-minute period, during which $80$ taxis traverse the area.

We deploy $15$ ESs evenly along the roads to provide MEC services. Each ES manages a 5G network with $270$ BUs (i.e., $B_k=270$), delivering a total data rate of $37$ MB/s (obtained from profiling) under optimal channel quality. Additionally, each ES is equipped with a GPU server containing $20$ CUs (i.e., $C_k=20$), where the GPU type is randomly selected from three models (RTX3090, RTX4090, RTX4500). Notably, for the same allocated CUs, the execution time of a given service varies across different GPU types. We consider four image classification services (ResNet-101/152 \cite{he2016identity}, Convnext-L \cite{trockman2022patches}, RegNet-Y-L \cite{radosavovic2020designing}), and profile their execution on each GPU type.  Furthermore, the service initialization time on various ESs is randomly sampled from the range of $[10, 50]$ ms. 

Each vehicle runs multiple applications (tasks). For each task $n_i$, its input data is randomly sampled from the ImageNet ILSVRC dataset \cite{ILSVRC15}, with data sizes ranging from $0.14$MB to $0.6$MB. The local application model is randomly selected from three lightweight image classification networks: ResNet-18 \cite{he2016identity}, GoogleNet \cite{szegedy2015going}, and RegNet-X-S \cite{radosavovic2020designing}. We profile the Nvidia Jetson Nano, an embedded GPU-based device, to obtain local execution time $d_i^\text{loc}$ and accuracy for each local model. The period is set to $60$, $70$, or $80$ms for tasks using model ResNet-18, GoogleNet, and RegNet-X-S, respectively.

\textit{Benchmark Algorithms}. We compare $\mathtt{FastSA}$ with three related algorithms adapted from prior studies: $\mathtt{Greedy}$ \cite{gao2022deadline}, $\mathtt{Game}$ \cite{li2019cooperative}, and $\mathtt{Graph}$ \cite{gao2025energy}. Among them, only $\mathtt{Graph}$ provides a theoretical guarantee of $\frac{1-\alpha}{2}$, where $\alpha$ denotes the maximum fraction of an ES's capacity that can be allocated to a single task. In contrast, $\mathtt{Greedy}$ and $\mathtt{Game}$ offer no such theoretical guarantee for problem $\mathbf{P}$. Since $\mathtt{FastSA}$, $\mathtt{Greedy}$, and $\mathtt{Game}$ do not require $\alpha$, we set $\alpha=\frac{1}{2}$ for $\mathtt{Graph}$ in our experiment.
\begin{itemize}
    \item \textbf{$\mathtt{Greedy}$}: For each $\ell \in \mathcal{L}$, $\mathtt{Greedy}$ defines resource efficiency as $e(\ell) = u(\ell)/(\frac{b_{\ell}}{B_j} \times \frac{c_\ell}{C_k})$, where $m_j$ and $m_k$ are the offload-ES and process-ES of $\ell$. It then processes all $\ell \in \mathcal{L}$ in non-increasing order of $e(\ell)$, selecting each service whenever feasibility constraints are met.
    \item \textbf{$\mathtt{Game}$}: $\mathtt{Game}$ models $\mathbf{P}$ as a non-cooperative game, in which each vehicle acts as a player, and selects the $\ell \in \mathcal{L}$ that maximizes total utility in each iteration.
    \item \textbf{$\mathtt{Graph}$}: $\mathtt{Graph}$ converts $\mathbf{P}$ into a tripartite graph matching problem and solves it with an LP rounding method.
\end{itemize}

\addtolength{\dbltextfloatsep}{-3mm}
\addtolength{\abovecaptionskip}{-2mm}

\begin{figure*}[t]
    \centering
    \subfigure[]{\includegraphics[width=0.32\linewidth]{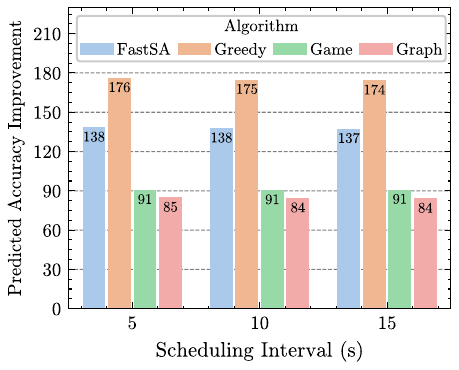} \label{fig:exp-normal-1}} 
    \subfigure[]{\includegraphics[width=0.32\linewidth]{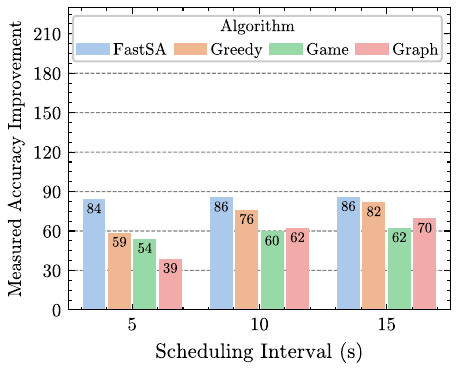} \label{fig:exp-normal-2}}
    \subfigure[]{\includegraphics[width=0.32\linewidth]{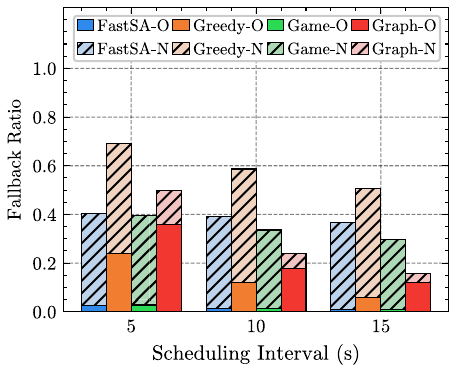} \label{fig:exp-normal-3}} 
    \caption{In fault-free MEC, for different scheduling intervals, the average (a) \textit{predicted} and (b) \textit{measured} accuracy improvements, and (c) fallback ratio.}
    \label{fig:exp-normal-a}
\end{figure*}

\begin{figure*}[t]
    \centering
    \subfigure[]{\includegraphics[width=0.4\linewidth]{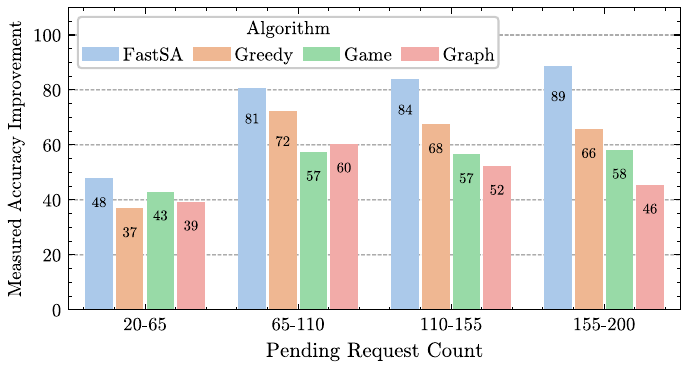} \label{fig:exp-normal-4}} 
    \hspace{1cm}
    \subfigure[]{\includegraphics[width=0.4\linewidth]{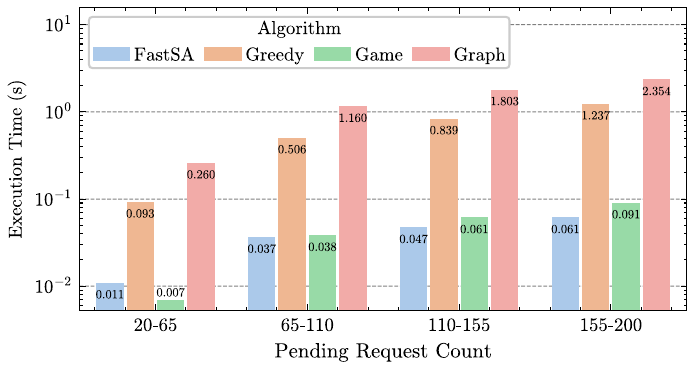} \label{fig:exp-normal-5}}
    \caption{In fault-free MEC, for different request counts, the results for (a) measured accuracy improvement and (b) algorithm execution time.}
    \label{fig:exp-normal-b}
\end{figure*}

Scheduling algorithm performance is assessed using two metrics: predicted and measured accuracy improvement. That is, the improved image classification accuracy as a result of successfully offloading the task and receiving the results within the deadline when compared to executing the task only locally on the ED. The \textit{predicted accuracy improvement} represents the result estimated by the algorithm itself, assuming static network conditions for the whole scheduling interval and neglecting execution overhead. In contrast, the \textit{ measured accuracy improvement} is obtained from simulation results based on the actual offloaded job instances, capturing the impact of network fluctuations and algorithm execution overhead. Since the simulation lasts $900$ seconds, the results are reported as the total accuracy improvement per second for consistency.

\begin{figure}[t]
    \centering
    \subfigure[]{\includegraphics[width=0.48\linewidth]{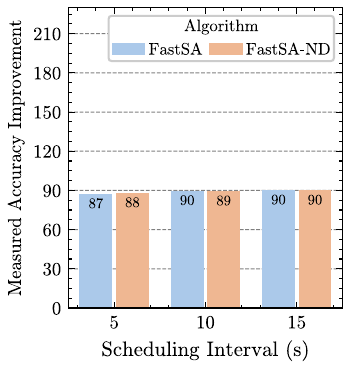} \label{fig:exp-ablation-2}}
    \subfigure[]{\includegraphics[width=0.48\linewidth]{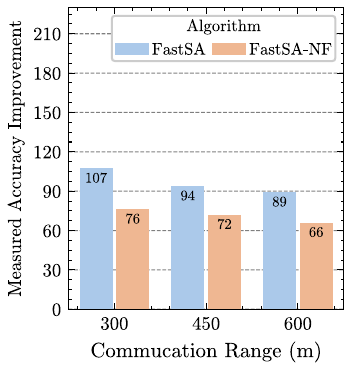} \label{fig:exp-ablation-4}}
    \caption{Ablation results for $\mathtt{FastSA}$ in fault-free MEC. (a) The impact of service-instance division under varying scheduling intervals; (b) the impact of task forwarding under varying effective communication ranges. Measured accuracy improvements are shown.}
    \label{fig:exp-ablation}
\end{figure}

\subsection{Result Discussion}\label{subsec:experiment-result}
Two sets of experiments are conducted. Since the scheduling algorithm is invoked once for each scheduling interval, the first experiment evaluates $\mathtt{FastSA}$ in a fault-free MEC under different scheduling intervals. The second experiment introduces backhaul link and ES failures during simulation to assess $\mathtt{FastSA}$ under our fault detection and recovery framework.

\subsubsection{Fault-Free MEC}
In the first experiment, we evaluate $\mathtt{FastSA}$ in a fault-free MEC under different scheduling intervals: $5$s, $10$s, and $15$s. The average predicted and measured accuracy improvements are presented in Figs.~\ref{fig:exp-normal-1} and \ref{fig:exp-normal-2}, respectively. On average, the predicted accuracy improvement of $\mathtt{FastSA}$ is $27\%$ lower than that of $\mathtt{Greedy}$, but $34.2\%$ and $38.7\%$ higher than those of $\mathtt{Game}$ and $\mathtt{Graph}$, respectively. Besides, the measured accuracy improvement of $\mathtt{FastSA}$ surpasses $\mathtt{Greedy}$, $\mathtt{Game}$, and $\mathtt{Graph}$ by $15.3\%$, $31.1\%$, and $33.3\%$, respectively. Although $\mathtt{FastSA}$ achieves the second-highest predicted accuracy improvement, its measured performance significantly exceeds that of all benchmark algorithms, demonstrating its superior resource utilization efficiency in realistic, dynamic MEC environments.

\begin{figure*}[t]
    \centering
    \subfigure[]{\includegraphics[width=0.4\linewidth]{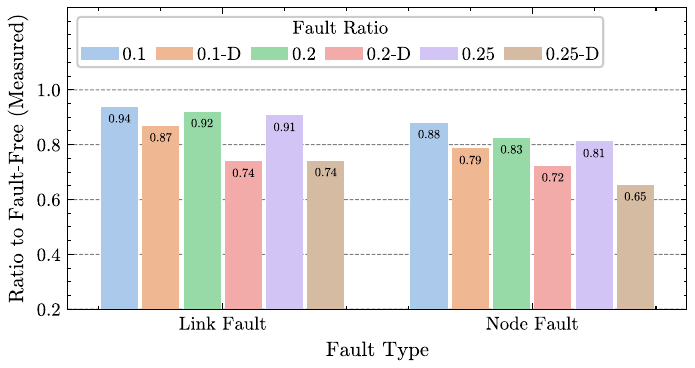} \label{fig:exp-fault-1}} 
    \hspace{1cm}
    \subfigure[]{\includegraphics[width=0.4\linewidth]{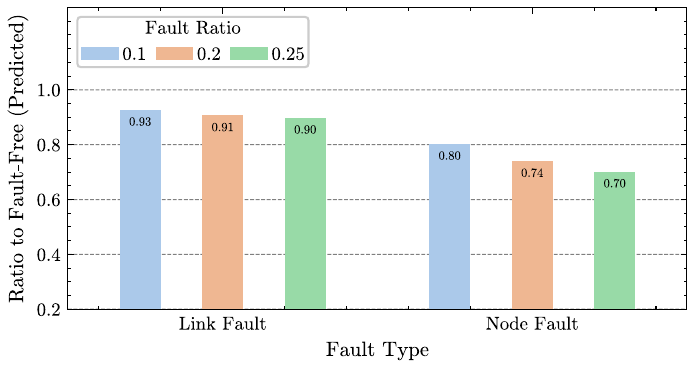} \label{fig:exp-fault-2}}
    \caption{In fault-injected MEC, the normalized average (a) measured and (b) predicted accuracy improvement for different fault types and fault ratios.}
    \label{fig:exp-fault}
\end{figure*}

A comparison between Figs.~\ref{fig:exp-normal-1} and \ref{fig:exp-normal-2} reveals a clear discrepancy between the predicted and measured accuracy improvements. On average, the measured accuracy improvements of $\mathtt{FastSA}$, $\mathtt{Greedy}$, $\mathtt{Game}$, and $\mathtt{Graph}$ are $38.1\%$, $58.7\%$, $35.2\%$, and $32.6\%$ lower than their predicted values, respectively. The underlying cause is shown in Fig.~\ref{fig:exp-normal-3}, which illustrates the average fallback ratio in each scheduling interval. The \textit{fallback ratio} is defined as the proportion of jobs successfully served (i.e., offloaded and completed within deadlines) relative to the total number of generated jobs for tasks that are granted in each scheduling interval. It consists of two components: ``-O'' in Fig.~\ref{fig:exp-normal-3} denotes jobs not offloaded due to scheduling algorithm overhead, while ``-N'' corresponds to those not offloaded due to network quality fluctuations (the channel quality decides the MCS selection, and hence the data offloading rate). As shown, the average fallback ratios of $\mathtt{FastSA}$, $\mathtt{Greedy}$, $\mathtt{Game}$, and $\mathtt{Graph}$ are $38.7\%$, $59.5\%$, $34.3\%$, and $29.9\%$, respectively. Moreover, the fallback ratios of $\mathtt{FastSA}$, $\mathtt{Greedy}$, and $\mathtt{Game}$ are primarily influenced by network variations, whereas that of $\mathtt{Graph}$ is dominated by algorithmic overhead. 
Besides, the $38.7\%$ fallback ratio of $\mathtt{FastSA}$ highlights a promising research direction: incorporating wireless network variation awareness into the problem formulation to enhance robustness against such dynamics.

We further examine how the measured accuracy improvement and execution time vary with the number of pending requests, as shown in Figs.~\ref{fig:exp-normal-4} and \ref{fig:exp-normal-5}. In Fig.~\ref{fig:exp-normal-4}, the measured accuracy improvement of all algorithms generally increases with the request count, except for $\mathtt{Graph}$, whose performance degrades due to its scheduling overhead (and hence a higher fallback ratio). $\mathtt{FastSA}$ consistently outperforms $\mathtt{Greedy}$, $\mathtt{Game}$, and $\mathtt{Graph}$, with the performance gap widening as the number of requests grows. Meanwhile, the average execution time (Fig.~\ref{fig:exp-normal-5}) of $\mathtt{FastSA}$ is $0.05$s, which is $18.4\times$, $1.4\times$, and $37.5\times$ faster than $\mathtt{Greedy}$, $\mathtt{Game}$, and $\mathtt{Graph}$, respectively. These results demonstrate the superior practical performance and runtime efficiency of $\mathtt{FastSA}$.

\subsubsection{Ablation Evaluation in Fault-Free MEC} In the second experiment, we evaluate the impact of service-instance division and task forwarding on $\mathtt{FastSA}$ in a fault-free MEC setting. The average measured accuracy improvement under different scheduling intervals is shown in Fig.~\ref{fig:exp-ablation-2}. $\mathtt{FastSA}$-$\mathtt{ND}$ denotes a variant of $\mathtt{FastSA}$ that processes all service instances in a unified category order (LL, LH, HL, HH) without division. As shown, $\mathtt{FastSA}$ and $\mathtt{FastSA}$-$\mathtt{ND}$ achieve similar measured accuracy improvements across all scheduling intervals, indicating that instance division does not incur noticeable performance loss. Meanwhile, $\mathtt{FastSA}$ additionally provides a constant-factor approximation guarantee, improving theoretical robustness without sacrificing practical performance.

We next evaluate the impact of task forwarding. Since 5G coverage is spatially limited, enabling forwarding increases the flexibility of task placement and allows the system to better utilize ES resources. Here, we consider different effective gNB communication ranges, where a task may only be offloaded to a gNB if its ED lies within the gNB’s communication range. The measured accuracy improvement under different effective gNB communication ranges is shown in Fig.~\ref{fig:exp-ablation-4}, where $\mathtt{FastSA}$-$\mathtt{NF}$ disables task forwarding. On average, $\mathtt{FastSA}$ achieves $26.2\%$ higher measured accuracy improvements than $\mathtt{FastSA}$-$\mathtt{NF}$, indicating that task forwarding effectively improves performance by enabling more flexible use of ES resources.

\subsubsection{Fault-Injected MEC} In the third experiment, we evaluate the accuracy improvement of $\mathtt{FastSA}$ under the proposed fault detection and recovery framework by manually injecting backhaul link faults (termed \textit{link fault}) and ES faults (termed \textit{node fault}). The simulation runs for $900$s with a scheduling interval of $10$s. Three fault ratios are considered for both link and node faults: $0.1$, $0.2$, and $0.25$. For each ratio, randomly sampled link or node faults are injected into the MEC in the middle of every $100$s time interval and are removed at the end of the same interval, e.g., during the time interval $[0,100]$, a fault is injected at $50$s and cleared at $100$s. 

In the simulation, the period of the HELLO message for link state monitoring is set to $100$ms. The recovery time is defined as the duration between fault injection and the election of a new global scheduler. During this interval, all tasks are processed locally.
Simulation results show that the recovery time does not exceed $108$ms, which is negligible compared to the $10$s scheduling interval. This demonstrates the superior efficiency of our framework in fault detection and recovery.

The average measured accuracy improvement is presented in Fig.~\ref{fig:exp-fault-1}, where all values are normalized to the fault-free baseline. For comparison, results with the fault detection and recovery framework disabled are also included (denoted as ``-D'' in Fig.~\ref{fig:exp-fault-1}). When the framework is disabled, the ES hosting the global scheduler is excluded from fault injection; otherwise, offloading would not be possible. As shown in Fig.~\ref{fig:exp-fault-1}, under link faults, the average measured accuracy improvements with and without the framework drop to $92.2\%$ and $78.3\%$ of the fault-free baseline, respectively. With the framework enabled, performance remains close to the fault-free level and only slightly declines as the fault ratio increases. This is because link failures do not reduce the total ES resources in the MEC, and the backhaul network can re-route data through alternative paths after resynchronization. In contrast, disabling the framework leads to a significant drop in accuracy improvement as the fault ratio grows, as the scheduler is not aware of the infeasible backhaul paths, and deadlines will be missed for job instances routed through those paths.

For node faults, as shown in Fig.~\ref{fig:exp-fault-1}, the measured accuracy improvement decreases with higher fault ratios, with a notably greater degradation when the framework is disabled. On average, the measured accuracy improvements with and without the framework are $83.9\%$ and $72.1\%$ of the fault-free baseline, respectively. Node faults result in more severe performance degradation than link faults because they directly reduce the available ES resources in the MEC. These results verify the effectiveness of our framework in maintaining system performance under both link and node faults.

The average predicted accuracy improvement is presented in Fig.~\ref{fig:exp-fault-2}, where all results are normalized to the fault-free baseline. On average, the predicted accuracy improvement decreases to $91\%$ and $74.8\%$ of the fault-free baseline under link and node faults, respectively. 
A comparison between Figs.~\ref{fig:exp-fault-1} and \ref{fig:exp-fault-2} indicates that the reduction in measured accuracy improvement is smaller than that in the predicted values, particularly under node faults. This discrepancy arises because, after a fault is injected, fewer tasks are granted for offloading. Consequently, network interference due to concurrent data offloading is reduced, leading to a smaller fallback ratio caused by network quality fluctuations.


\section{Conclusion}\label{sec:conclusion}
This paper investigated a forwarding-enabled DOAP in MEC, denoted as $\mathbf{P}$, which jointly optimizes task offloading and resource allocation under wireless bandwidth and computational resource constraints to maximize the total inference accuracy improvement. To solve $\mathbf{P}$, we proposed $\mathtt{FastSA}$, a linear-time $\frac{1}{10}$-approximation algorithm that offers the first constant-factor guarantee for the forwarding-enabled DOAP. We also designed a fault-tolerant task offloading framework capable of detecting and recovering from backhaul link and ES failures. Experimental results show that $\mathtt{FastSA}$ outperforms all benchmark algorithms in fault-free settings while achieving the lowest execution time, and that the proposed fault-handling framework effectively preserves service quality under fault-injected scenarios. Lastly, we observed a gap between the predicted and measured accuracy improvements of $\mathtt{FastSA}$ due to wireless network quality fluctuations. As future work, we plan to incorporate wireless variation awareness into the problem formulation to enhance robustness against such dynamics.

\newpage

\bibliographystyle{IEEEtran}
\bibliography{IEEEabrv,reference}

\end{document}